\documentclass[11pt]{article}
\usepackage{amssymb, graphicx}
\newtheorem{theorem}{Theorem}
\newtheorem{lemma}{Lemma}

\newtheorem{corollary}{Corollary}
\newtheorem{remark}{Remark}

\newenvironment{proof}
{\begin{trivlist} \item[]{\bf Proof. }}%
{\hspace*{\fill}$\rule{.3\baselineskip}{.35\baselineskip}$\end{trivlist}}

\newcommand{\R}{\mathbb{R}}

\newcommand{\Z}{\mathbb{Z}}
\newcommand{\N}{\mathbb{N}}

\begin{document}

\title{\bf Breather continuation from infinity \\ in nonlinear oscillator chains}

\author{Guillaume James$^{a}$ and Dmitry Pelinovsky$^{b}$ \\
{\small $^a$ Laboratoire Jean Kuntzmann, UMR CNRS 5224,} \\
{\small BP 53, 38041 Grenoble Cedex 9, France}\\
{\small $^b$ Department of Mathematics and Statistics, McMaster University, } \\
{\small Hamilton, Ontario, Canada, L8S 4K1}
}

\date{\today}
\maketitle


\begin{abstract}
Existence of large-amplitude time-periodic breathers localized near
a single site is proved for the discrete Klein--Gordon equation, in the case
when the derivative of the on-site potential has a compact support.
Breathers are obtained at small coupling
between oscillators and under nonresonance conditions.
Our method is different from the classical anti-continuum limit
developed by MacKay and Aubry, and yields in general
branches of breather solutions that cannot
be captured with this approach.
When the coupling constant goes to zero,
the amplitude and period of oscillations
at the excited site go to infinity.
Our method is based on near-identity transformations, analysis of
singular limits in nonlinear oscillator equations, and
fixed-point arguments.
\end{abstract}

\section{Introduction}

Recent studies of spatially localized and time-periodic oscillations
(breathers) in lattice models of DNA
\cite{James3,James4} call for systematic analysis
of such excitations in the discrete
Klein--Gordon equation
\begin{equation}
\label{KGlattice} \ddot{x}_n + V'(x_n) = \gamma \left( x_{n+1} - 2
x_n + x_{n-1} \right), \quad n \in \mathbb{Z},
\end{equation}
where $\gamma >0$ is a coupling constant, $V :
\mathbb{R} \to \mathbb{R}$ is a nonlinear potential, and
${\bf x}(t)=\{x_n(t) \}_{n \in \mathbb{Z}}$ is a sequence of real-valued amplitudes
at time $t \in \mathbb{R}$.

In the classical Peyrard-Bishop model for DNA \cite{pey04},
$V$ is a Morse potential having a global minimum at
$x = 0$, confining as $x \to -\infty$ and saturating at a constant level as
$x \to \infty$. However, recent studies \cite{weber,James1,James3}
argued that the Morse potential should be replaced by a potential with
a local maximum at $x = a_0 > 0$, which induces a
double-well structure, where one of the wells extends to infinity
(both kinds of potentials are depicted in Figure \ref{figpot}).
The existence of breathers residing in the
potential well near $x = 0$ can be proved with classical methods
such as the center manifold reduction for maps \cite{james,James2},
variational methods \cite{aubkadel,pankovb}, and the continuation
from the anticontinuum limit $\gamma \rightarrow 0$ \cite{Aubry2,MA94,sepmac}.

\begin{figure}
\begin{center}
\includegraphics[width = 10cm]{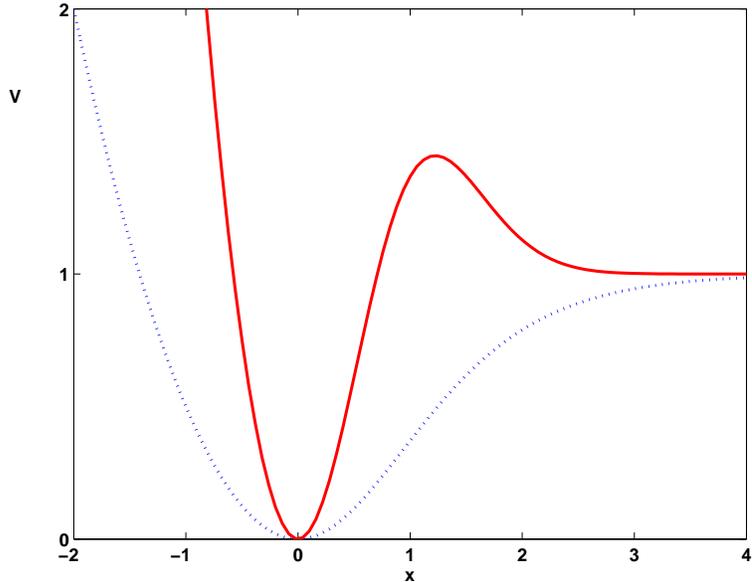}
\caption{\label{figpot}Morse potential (dashed line) and modified double-well potential (full line).}
\end{center}
\end{figure}

A more delicate problem is the existence of
large-amplitude breathers residing in the other
potential well which extends to infinity.
Large-amplitude stationary
solutions bifurcating from infinity as
$\gamma \to 0$ have been obtained in
\cite{James3}. These solutions are localized near a single site,
say $n = 0$, and their amplitude diverges
as $\gamma \to 0$. Large-amplitude breathers
in a finite-size neighborhood of these stationary solutions have been
constructed in \cite{James4} for small coupling $\gamma$, using the contraction mapping
theorem and scaling techniques.
These large-amplitude breathers
oscillate beyond the potential barrier of $V$ at $x = a_0$,
and their amplitude goes to infinity
as $\gamma \to 0$. Existence of
large-amplitude breathers oscillating everywhere above the potential barrier
of $V$ was left open in \cite{James4}.

Our goal is to show the existence of large-amplitude breathers
oscillating in several potential wells, setting-up a continuation of these
solutions from infinity as $\gamma \to 0$.
To illustrate some key points of our analysis,
let us consider the example
\begin{equation}
\label{potential-example}
V(x) = \frac{1}{4} (1 + e^{-x^2}(x^2-1)).
\end{equation}
Here $V$ has a global minimum at $x = 0$, a pair of symmetric global
maxima at $x = \pm a_0$ with $a_0 > 0$, and $\lim_{x\rightarrow\pm\infty}V(x) = \frac{1}{4}$.

In the standard anti-continuum limit, one sets $\gamma = 0$ and
$x_n = 0$ for all $n \in \mathbb{Z} \backslash \{0\}$, and one considers
a time-periodic solution $x_0(t) \equiv x(t)$ of the nonlinear oscillator equation
\begin{equation}
\label{oscil}
\ddot{x} + V'(x) = 0.
\end{equation}
Under a nonresonance condition, this compactly supported time-periodic solution can be continued
for $\gamma \approx 0$ into an exponentially localized time-periodic breather solution using the implicit function theorem  \cite{MA94}.

The phase plane $(x,\dot{x})$ and the
frequency-amplitude $(\omega,a)$ diagram of the nonlinear oscillator equation
(\ref{oscil}) with the potential (\ref{potential-example})
are shown on Figure \ref{fig2}. In this case, the periodic
solution $x(t)$ has a cut-off amplitude at $a = a_0$. Only the family
of periodic solutions with $a \in (0,a_0)$ can be continued by the
anti-continuum technique developed by MacKay and Aubry \cite{MA94}.

In addition, there are two families of unbounded solutions: one corresponds
to oscillations beyond the potential barrier
of $V$ for $|x| > a_0$ and the other one corresponds to oscillations above the potential barrier.
Roughly speaking, the new technique developed in  \cite{James4} allows one to obtain
large amplitude breathers ``close" to unbounded solutions of the first family for $\gamma \approx 0$.

The present paper considers large-amplitude breathers near the second family of
unbounded solutions.
These two families of breathers are obtained by ``continuation from infinity"
for arbitrarily small values of $\gamma$, but without reaching $\gamma = 0$.
In this case, the potential $V$ in the nonlinear oscillator equation
(\ref{oscil}) can be simply replaced by
\begin{equation}
\label{potential-new}
V_{\gamma}(x) = V(x) + \gamma x^2.
\end{equation}
The potential $V_{\gamma}$ includes a restoring force originating
from the nearest-neighbors coupling in the discrete Klein--Gordon equation (\ref{KGlattice}).
As $\gamma\rightarrow 0$, the amplitudes and periods of the resulting breathers
go to infinity. As a result, we need a careful control of nonresonance
conditions in order to prove the existence of such breathers.

Although a part of our continuation procedure involving the contraction mapping theorem
is close to the one developed in \cite{James4}, our mathematical
analysis is quite different because our breather solutions scale differently
in the different potential wells, which induces some singular
perturbation analysis and more delicate estimates than in \cite{James4}.
Note also that the contraction mapping theorem has been used by
Treschev \cite{Treschev} to prove the existence of other types
of localized solutions (solitary waves) in Fermi-Pasta-Ulam lattices,
in which nearest-neighbors are coupled by an anharmonic potential having a repulsive
singularity at a short distance. In this case, the existence problem yields an
advance-delay differential equation with other kinds of
mathematical difficulties.

\begin{figure}
\begin{center}
\includegraphics[width = 8cm]{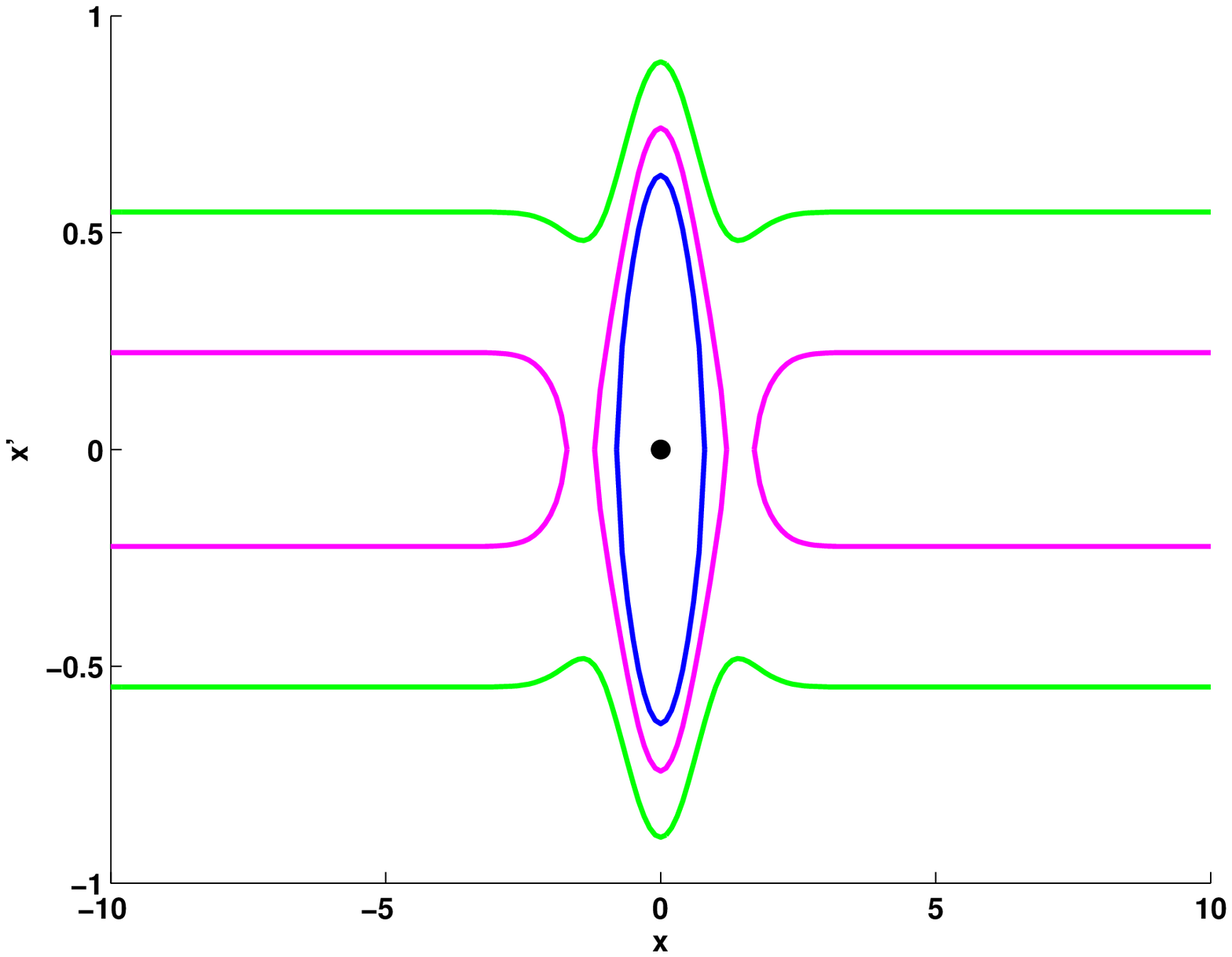}
\includegraphics[width = 8cm]{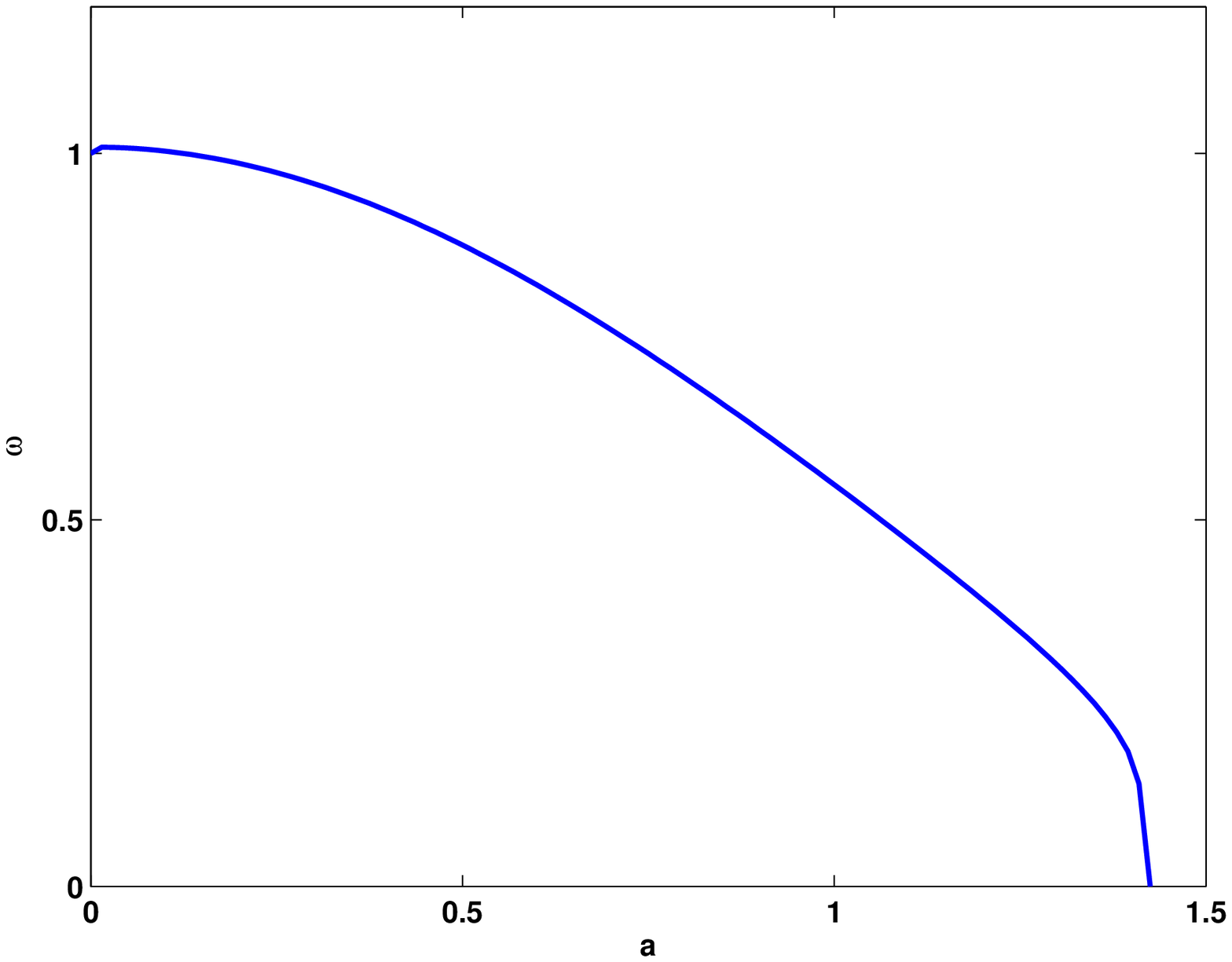}
\caption{\label{fig2} The phase plane $(x,\dot{x})$ (left) and the frequency--amplitude diagram $(\omega,a)$ (right)
for the potential (\ref{potential-example}).}
\end{center}
\end{figure}

To simplify our analysis, we assume that $V$ is symmetric and bounded, whereas
$V^\prime$ has a compact support. To be precise, the following properties on $V$ are assumed:
\begin{itemize}
\item[P1] $V \in C^2(\mathbb{R})$ and $V(-x) = V(x)$ for all $x \in \R$;

\item[P2] There is $x_0 > 0$ such that $V \in C^7(-x_0,x_0)$
and the Taylor expansion of $V$ at $x = 0$ is
$V(x) = \frac{1}{2} \kappa^2 x^2 + {\cal O}(x^6)$  with $\kappa > 0$;

\item[P3] $0 \leq V(x) \leq V_L$ for all $x \in \mathbb{R}$ and some $V_L > 0$;

\item[P4] $V'(x)$ is compactly supported on $[-a_0,a_0]$ for some
$a_0 > 0$ such that \\
$V(x) = V_{\infty}$ for $|x| \geq a_0$ and
$V_{\infty} \in (0,V_L]$.
\end{itemize}

Assumption (P1) allows us to consider symmetric periodic oscillations, which can be studied on
the quarter of the fundamental period.
This assumption simplifies the presentation but is not essential, and our analysis
could be extended e.g. to potentials confining at $-\infty$ (as in Figure \ref{figpot}).

Assumption (P2) allows us to develop a contraction mapping
argument for the small-amplitude oscillations on the sites $n \neq 0$,
a procedure which cannot be carried out if a quartic term is present
in the expansion of $V$ near the origin.
It would be useful to relax this condition, which assumes a very weak
anharmonicity of small amplitude oscillations. Note that the quartic term in $V(x)$ near $x = 0$ is also
excluded in the recent analysis of scattering of small initial data to
zero equilibrium by Mielke \& Patz \cite{Mielke}.

Assumption (P3) allows for
large-amplitude oscillations at the central site $n = 0$.

Assumption (P4) allows us to consider
linear oscillations of the central site outside the compact support of $V'$.
This property is used in Lemma \ref{lemma-spectrum} below to solve the singularly perturbed
oscillator equation for renormalized oscillations at $n=0$. This compact support assumption
is quite restrictive, and it would be interesting to relax it in a future work, by
considering e.g. exponentially decaying potentials (as in example (\ref{potential-example}))
and treating exponential tails as perturbations of the present case.

We note that Fura \& Rybicki \cite{Fura} have
proved the existence of periodic solutions bifurcating from
infinity for a class of finite-dimensional Hamiltonian systems
with asymptotically linear potentials using degree theory. Our analysis is
different and consists in two steps. We first reduce the infinite-dimensional
Hamiltonian system to a perturbed oscillator equation describing large amplitude
oscillations at the breather center, using the contraction mapping theorem.
Once this has been achieved, we solve the reduced problem using
a topological method (Schauder's fixed point theorem).

Our main result is the existence of the large-amplitude breathers
if the potential $V$ satisfies assumptions (P1)--(P4) as well as the technical
non-degeneracy condition in equation (\ref{condition-central-nonresonance}) below.
As further problems, it would be
interesting to analyze the existence of multibreather solutions
bifurcating from infinity, as well as the stability of such solutions,
as it was done previously for finite-amplitude breathers near
the standard anti-continuum limit
(see, e.g., \cite{macsep,Bambusi,Aubry2,marinstab,archistab,koukou}).

The article is organized as follows. Section 2 describes the main results.
Large-amplitude oscillations near $n = 0$ are analyzed in Section 3.
Small-amplitude oscillations for $n \neq 0$ are considered in Section 4.
The proof of the main theorem is given in Section 5. Section 6 gives a proof
that the large-amplitude breather decays exponentially in $n \in \Z$.

{\bf Acknowledgement.} This work was initiated during the visit of
D.P. to Laboratoire Jean Kuntzmann supported in part by the
Ambassade de France au Canada. D.P. thanks the members of the
Laboratory for hospitality during his visit.

\section{Main results}

We shall consider the discrete Klein--Gordon equation (\ref{KGlattice}) for
small $\gamma > 0$ and assume that the breather is localized near the central
site $n = 0$.
We consider oscillations
in the potential $V_{\gamma}(x)$ at the energy level $E$:
\begin{equation}
\label{newton-particle}
\ddot{x} + V_{\gamma}'(x) = 0 \quad \Rightarrow \quad E = \frac{1}{2} \dot{x}^2 + V_{\gamma}(x).
\end{equation}
Thanks to assumption (P3), the anti-continuum limit $\gamma \to 0$ is singular
for $E > V_L$ in the sense that a bounded trajectory of system
(\ref{newton-particle}) trapped by the quadratic potential $\gamma
x^2$ degenerates into an unbounded trajectory as $\gamma \to 0$.

We would like to select a unique $T$-periodic solution of (\ref{newton-particle})
by fixing its energy $E > V_L$ and choosing $\gamma$ small enough.
For a fixed $E > V_L$, we will be working for sufficiently small
$\gamma > 0$ to ensure that $V_{\gamma}(a) = E$ admits a unique
positive solution $a(E,\gamma )$. More precisely, thanks to assumptions (P3) and (P4),
we obtain a unique solution $a = (E - V_{\infty})^{1/2} \gamma^{-1/2}$
for $\gamma < (E-V_L)/a_0^2$.
Fixing $\dot{x}(0) = 0$, we can parameterize periodic solutions
by $x(0) = a(E,\gamma) > 0$, and their period can be written $T = T(E,\gamma)$.
Thanks to assumption (P4), we shall prove (in Section \ref{section-2}) that for any $E > V_L$
\begin{equation}
\label{period-expansion-intro} T(E,\gamma) = \frac{\sqrt{2} \pi}{\gamma^{1/2}}
+ \lambda(E) + {\cal O}(\gamma) \quad \mbox{\rm as} \quad \gamma \to 0
\end{equation}
where
$$
\lambda(E) = 2 \sqrt{2} \left( \int_0^{a_0} \frac{dx}{(E - V(x))^{1/2}} - \frac{a_0}{(E - V_{\infty})^{1/2}} \right).
$$

Thanks to assumption (P1), the $T$-periodic solution $x(t)$ of the nonlinear oscillator equation
(\ref{newton-particle}) with $x(0) = a > 0$ and $\dot{x}(0) = 0$
is symmetric with respect to reflections about the points $t = 0$ and $t = \frac{T}{2}$
and anti-symmetric with respect to reflection about the points $t = \frac{T}{4}$ and $t = \frac{3 T}{4}$.
Therefore, the $T$-periodic solution satisfies
\begin{equation}
\label{symmetry}
x(-t) = x(t) = -x\left(\frac{T}{2}-t \right) \quad t \in \R.
\end{equation}

The normalized frequency of oscillations is defined by
\begin{equation}
\label{defw0}
\omega_0(E,\gamma) = \frac{2\pi}{T(E,\gamma)} \gamma^{-1/2}
\end{equation}
such that $\omega_0(E,\gamma) \to \sqrt{2}$ as $\gamma \to 0$ for a fixed $E > V_L$.
To avoid resonances of large-amplitude oscillations
at the central site $n = 0$ with small-amplitude oscillations at the other sites $n \in \mathbb{Z} \backslash \{0\}$,
we will show (in Section \ref{section-3}) that the following non-resonance conditions
\begin{equation}
\label{resonances} \kappa^2 - m^2 \gamma \omega_0^2(E,\gamma) + 2 \gamma (1 -
\cos(q)) \neq 0,
\end{equation}
must be satisfied for all $m \in \mathbb{Z}$ and all $q \in [-\pi,\pi]$.

For a fixed $E > V_L$, we shall now consider the non-resonant set of parameters
(breather frequency $\omega = \frac{2\pi}{T}$ and coupling constant $\gamma$).
We plot on Figure \ref{fig3} the prohibited regions between the boundaries
of the non-resonant set, given by the curves
$\omega = \frac{\kappa}{m}$ and $\omega = \frac{\sqrt{\kappa^2
+ 4 \gamma}}{m}$,  together with the curve
$\omega = \sqrt{\gamma}\,  \omega_0(E,\gamma)$.
Non-resonance conditions (\ref{resonances}) are satisfied if $\gamma$ belongs to the set
$C_{E} = \cup_{m \geq m_0} (\Gamma_m,\gamma_m)$, where
$\gamma=\Gamma_m$ and $\gamma=\gamma_m$ correspond to
the intersections of the above curves starting with some $m_0 \geq 1$,
i.e. $\Gamma_m$ and $\gamma_m$ satisfy implicit equations
\begin{equation}
\label{Gamma-gamma}
\frac{\sqrt{\kappa^2
+ 4 \,\Gamma_m}}{m+1} = \sqrt{\Gamma_m} \omega_0(E,\Gamma_m), \quad
\frac{\kappa}{m}= \sqrt{\gamma_m} \omega_0(E,\gamma_m ), \quad m \geq m_0.
\end{equation}
Equations (\ref{Gamma-gamma}) can be solved for $m$ large enough
thanks to expansion (\ref{period-expansion-intro}) and the implicit function arguments,
yielding as $ m \to \infty$
$$
\Gamma_m =
\frac{\kappa^2}{2 m^2}\Big(1 + \frac{\kappa \lambda(E)}{\pi m} - \frac{2}{m}+ {\cal O}(m^{-2}) \Big),
\quad \gamma_m =
\frac{\kappa^2}{2 m^2}
\Big(1+\frac{\kappa \lambda(E)}{\pi m}+ {\cal O}(m^{-2}) \Big).
$$
In particular, we note that $\Gamma_m < \gamma_m$ and $|\gamma_m - \Gamma_m| = {\cal O}(m^{-3})$ as $m \to \infty$.

\begin{figure}
\begin{center}
\includegraphics[width = 10cm]{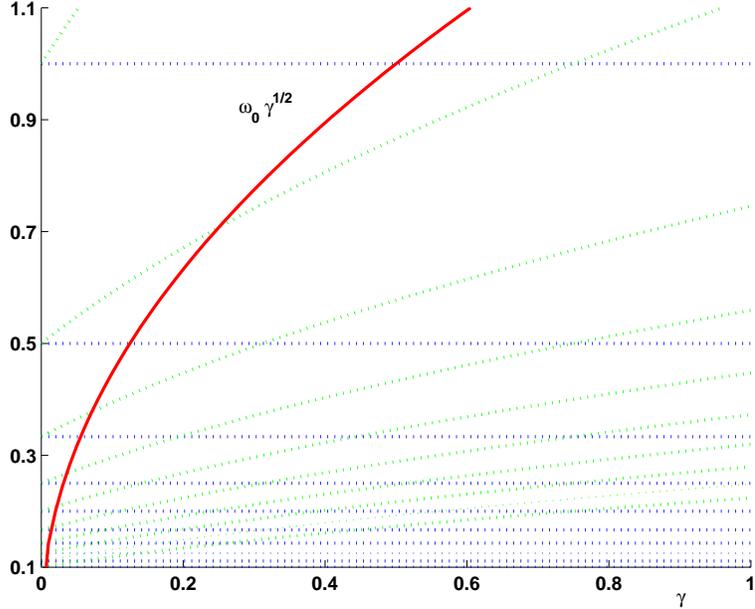}
\caption{\label{fig3}
Resonance tongues on the plane $(\gamma,\omega)$ (delimited by dotted lines),
and breather frequency curve $\omega = \omega_0 \gamma^{1/2}$ (solid line)
for $\kappa = 1$. No resonances occur for $\gamma \in \cup_{m \geq m_0} (\Gamma_m,\gamma_m)$,
where breather frequencies lie outside the resonance tongues.}
\end{center}
\end{figure}

We can now state the main result of this article. Note that the existence
of breathers is only obtained for a subset $\tilde{C}_{E,\nu} \subset C_E$
of the non-resonant values of the coupling constant $\gamma$, because our method
breaks down near the boundary of $C_E$.

\begin{theorem}
\label{theorem-main}
Assume (P1)--(P4) on $V(x)$ and fix $E > V_L$. Let $x(t)$ be a $T(E,\gamma )$-periodic solution
of the nonlinear oscillator equation (\ref{newton-particle}) for small $\gamma > 0$
satisfying symmetries (\ref{symmetry}) and assume that $\lambda'(E) \neq 0$, i.e.
\begin{equation}
\label{condition-central-nonresonance}
 \int_0^{a_0} \frac{dx}{(E -
V(x))^{3/2}} - \frac{a_0}{(E - V_{\infty})^{3/2}} \neq 0.
\end{equation}
Fix $\nu \in (0,1)$ and consider the set of coupling constants
$\tilde{C}_{E,\nu} = \cup_{m \geq m_0}
(\tilde{\Gamma}_m,\tilde{\gamma}_m) \subset C_E$, where
$\tilde{\Gamma}_m, \tilde{\gamma}_m$ are defined by the implicit equations
\begin{equation}
\label{Gamma-gamma-tildethm}
\frac{\sqrt{\kappa^2
+ 4 \, \tilde{\Gamma}_m}}{\sqrt{(m+1)^2 - \nu (m+1)}} =
\sqrt{\tilde{\Gamma}_m} \omega_0(E,\tilde{\Gamma}_m), \quad
\frac{\kappa}{\sqrt{m^2 + \nu m}}= \sqrt{\tilde{\gamma}_m} \omega_0(E,\tilde{\gamma}_m ),
\end{equation}
for $m\geq m_0$, and satisfy as $m\rightarrow +\infty$
$$
\tilde{\Gamma}_m =
\frac{\kappa^2}{2 m^2}\Big(1 + \frac{\kappa \lambda(E)}{\pi m} - \frac{2-\nu}{m}
+ {\cal O}(m^{-2}) \Big), \quad \tilde{\gamma}_m =
\frac{\kappa^2}{2 m^2}
\Big(1 + \frac{\kappa \lambda(E)}{\pi m} - \frac{\nu}{m}  + {\cal O}(m^{-2}) \Big).
$$
For all sufficiently small $\gamma$ in $\tilde{C}_{E,\nu}$,
there exists a $T$-periodic spatially localized
solution ${\bf x}(t) \in H^2_{\rm per}((0,T);l^2(\Z))$ of
the Klein--Gordon lattice (\ref{KGlattice}) such that
$$
x_n(t) = x_{-n}(t), \quad n \in \mathbb{Z}; \quad
{\bf x}(-t) = {\bf x}(t) = -{\bf x}\left(\frac{T}{2} - t\right), \quad  t \in \mathbb{R};
$$
and
\begin{eqnarray}
\label{bound-final}
\exists C > 0 : \quad \sup_{t \in [0,T]} | x_0(t) - x(t) | \leq C \gamma^{-1/4}, \quad
\sup_{n \geq 1} \sup_{t \in [0,T]} | x_n(t) | \leq C \gamma^{1/4}.
\end{eqnarray}
\end{theorem}

\begin{remark}
Since $\| x \|_{L^{\infty}} = a = {\cal O}(\gamma^{-1/2})$ as $\gamma \to 0$,
the first bound in (\ref{bound-final}) shows that the relative error
$\| x_0 - x\|_{L^{\infty}}/\| x \|_{L^{\infty}}$ is as small as
${\cal O}(\gamma^{1/4})$.
\end{remark}

\begin{remark}
If $\nu \in (0,1)$, we still have
$\tilde{\Gamma}_m < \tilde{\gamma}_m$ and $|\tilde{\gamma}_m - \tilde{\Gamma}_m| = {\cal O}(m^{-3})$
as $m \to \infty$. Therefore, the rate of decrease of the interval widths in the set $\tilde{C}_{E,\nu}$
corresponds to the rate of decrease of the widths of the non-resonant intervals in the set $C_E$.
\end{remark}

\begin{remark}
Although we do not attempt here to deal with non-compact
potentials, we believe that assumption (P4) can be
relaxed if $V'(x)$ has a sufficiently fast decay to zero as $|x| \to \infty$.
In that case, we conjecture that the non-resonance condition
(\ref{condition-central-nonresonance}) would be replaced by
$$
Q := \int_0^{\infty} \left[ \frac{1}{(E - V(x))^{3/2}} -
\frac{1}{(E - V_{\infty})^{3/2}} \right] dx \neq 0,
$$
where $V_{\infty} := \lim_{x \to \infty} V(x)$.
Figure \ref{fig4} illustrates that this condition is satisfied for the
particular potential (\ref{potential-example}),
for any finite $E > V_L$
(we note that the value of $Q$ approaches
$0$ as $E \to \infty$).
\end{remark}

Theorem \ref{theorem-main} is proved in Section 5, using intermediate results
established for the single oscillator equation (\ref{newton-particle}) with
a forcing term (in Section 3)
and for the discrete KG equation (\ref{KGlattice}) linearized at zero
equilibrium (in Section 4).
We finish the article with Section 6, where we prove that the amplitude of breather oscillations
decays exponentially in $n$ on $\mathbb{Z}$
in the following sense:
$$
\exists D_0 > 0 : \quad
\sup_{t \in [0,T]} | x_n(t) | \leq (D_0 \gamma)^{(2n-1)/4}, \quad n \geq 2,
$$
for all sufficiently small $\gamma  \in \tilde{C}_{E ,\nu}$.

\begin{figure}
\begin{center}
\includegraphics[width = 10cm]{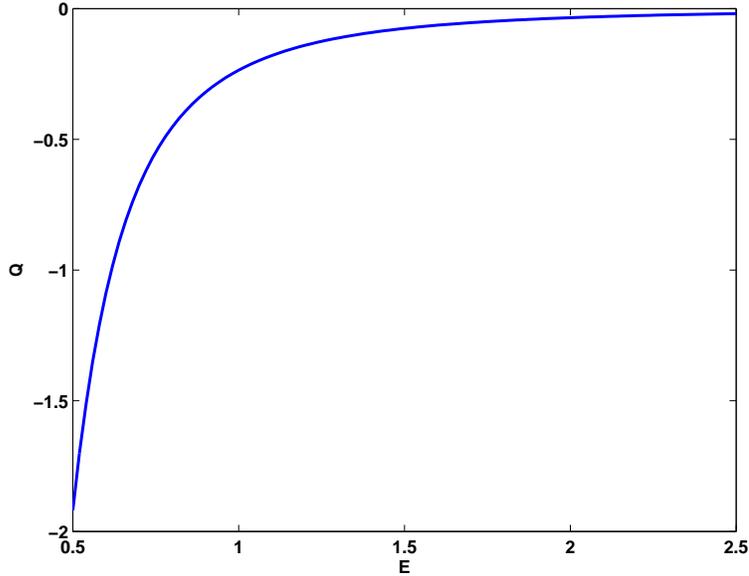}
\caption{\label{fig4} Non-resonance coefficient $Q$ versus energy $E$ for the potential (\ref{potential-example}).}
\end{center}
\end{figure}

\section{Large-amplitude oscillations at a central site}
\label{section-2}

We consider here solutions of the nonlinear oscillator equation (\ref{newton-particle})
in the singular limit $\gamma \to 0$. Assumptions (P1)--(P4) on the potential $V$ are used
everywhere, without further notes.

\begin{lemma}
Fix $E > V_L$.
There exists $\gamma_0 = \gamma_0(E) > 0$ such that for any $\gamma
\in (0,\gamma_0)$, there exist exactly two $T$-periodic solutions
of (\ref{newton-particle}) with amplitude $\| x \|_{L^{\infty}} = (E - V_{\infty})^{1/2}
\gamma^{-1/2}$ satisfying symmetries (\ref{symmetry}) and the asymptotic expansion
\begin{equation}
\label{period-expansion} T = \frac{\sqrt{2} \pi}{\gamma^{1/2}}
+ 2 \sqrt{2} \left( \int_0^{a_0} \frac{dx}{(E - V(x))^{1/2}} - \frac{a_0}{(E - V_{\infty})^{1/2}} \right)
+ {\cal O}(\gamma) \quad \mbox{\rm as} \quad \gamma \to 0.
\end{equation}
\label{lemma-oscillations}
\end{lemma}

\begin{proof}
Thanks to assumptions (P1) and (P3), for a fixed $E > V_L$ there exists $\gamma_0(E) > 0$
such that for any $\gamma \in (0,\gamma_0)$
equation $V_{\gamma}(x) = V(x) + \gamma x^2 = E$ admits only two solutions
$x = \pm a$ with $a > 0$ such that $V_{\gamma}'(x) > 0$ for all $x \geq a$. The two
periodic solutions with symmetries (\ref{symmetry}) are constructed from the same bounded trajectory on
the phase plane $(x,\dot{x})$ departing from either the point $(x(0),\dot{x}(0)) = (a,0)$
or the point $(x(0),\dot{x}(0)) = (-a,0)$.

Since $E = \gamma a^2 + V(a)$, it follows from assumption (P4) that
$$
a = (E - V_{\infty})^{1/2} \gamma^{-1/2} \quad \mbox{\rm as} \quad \gamma \to 0.
$$

Asymptotic expansion of the period $T$ of the periodic solution of (\ref{newton-particle})
is found from the exact formula
\begin{equation}
\label{formula-1}
T = \sqrt{2} \int_{-a}^{a} \frac{dx}{\sqrt{E -
V_{\gamma}(x)}} = 2 \sqrt{2} \int_{0}^{a} \frac{dx}{\sqrt{E -
\gamma x^2 - V(x)}}.
\end{equation}
Thanks to assumption (P4), we know that $V(x) = V_{\infty}$ for all
$x \in [a_0,a]$, so that
$$
T = \frac{2 \sqrt{2}}{\sqrt{\gamma}} \left( \frac{\pi}{2} - \theta_0 \right) +
2 \sqrt{2} \int_{0}^{a_0} \frac{dx}{\sqrt{E -
\gamma x^2 - V(x)}},
$$
where $\theta_0$ is the smallest positive root of $a \sin(\theta) = a_0$ satisfying
$$
\theta_0 = \arcsin\left( \frac{a_0 \gamma^{1/2}}{(E - V_{\infty})^{1/2}} \right) = \frac{a_0 \gamma^{1/2}}{(E - V_{\infty})^{1/2}} + {\cal O}(\gamma^{3/2}) \quad \mbox{\rm as} \quad \gamma \to 0.
$$
Since $E > V_L$ is fixed,
there is $C > 0$ such that $E - V(x) \geq C$ for all $x \in [0,a_0]$. As a result,
the asymptotic expansion
$$
\int_{0}^{a_0} \frac{dx}{\sqrt{E -
\gamma x^2 - V(x)}} = \int_{0}^{a_0} \frac{dx}{(E - V(x))^{1/2}} + {\cal O}(\gamma)
\quad \mbox{\rm as} \quad \gamma \to 0,
$$
concludes the proof of the asymptotic expansion (\ref{period-expansion}).
\end{proof}

Let us represent the solution of (\ref{newton-particle}) for $E > V_L$ in the form
\begin{equation}
\label{scaling-transformation}
x(t) = \frac{X(\tau)}{\gamma^{1/2}} , \quad \tau = \gamma^{1/2} t.
\end{equation}
By Lemma \ref{lemma-oscillations}, we have $\| X \|_{L^{\infty}}
= {\cal O}(1)$ and $T_0 := \gamma^{1/2} T = {\cal O}(1)$
as $\gamma \to 0$ with precise value $\| X \|_{L^{\infty}} = (E - V_{\infty})^{1/2}$ and
the asymptotic expansion
\begin{equation}
\label{period-expansion-T0} T_0 = \sqrt{2} \pi
+ 2 \sqrt{2} \gamma^{1/2} \left( \int_0^{a_0} \frac{dx}{(E - V(x))^{1/2}} - \frac{a_0}{(E - V_{\infty})^{1/2}} \right)
+ {\cal O}(\gamma^{3/2}) \quad \mbox{\rm as} \quad \gamma \to 0.
\end{equation}
We shall now derive a series of estimates that will be useful for the proof of Theorem \ref{theorem-main}.

\begin{corollary}
Let $X(\tau)$ be the $T_0$-periodic function defined by
the solution of Lemma \ref{lemma-oscillations}
in parametrization (\ref{scaling-transformation}) for any fixed $E > V_L$.
Then, $X \in C^3_{\rm per}(0,T_0)$
and for sufficiently small $\gamma > 0$,
there exists $C(E) > 0$ such that $\| X \|_{H^1_{\rm per}} \leq C(E)$.
Moreover, $\| X \|_{C^1} \leq (1+\sqrt{2})\, \sqrt{E}$.
\label{lemma-H2-solution}
\end{corollary}

\begin{proof}
We recall that $\| X \|_{L^{\infty}} = (E - V_{\infty})^{1/2} \leq \sqrt{E}$.
Let us rewrite the energy conservation (\ref{newton-particle}) in
parametrization (\ref{scaling-transformation}):
\begin{equation}
\label{energy-invariant}
\frac{1}{2} \dot{X}^2 + X^2 + V(\gamma^{-1/2} X) = E.
\end{equation}
Since $V\geq 0$ we have $\| \dot{X} \|_{L^{\infty}} \leq \sqrt{2E}$, which
gives the bound on $\|X \|_{C^1}$.
This also gives the uniform bound on $\|X \|_{H^1_{\rm per}}$
since $T_0 = {\cal O}(1)$ as $\gamma \to 0$.

Let us also rewrite the second-order equation (\ref{newton-particle})
in parametrization (\ref{scaling-transformation}):
\begin{equation}
\label{second-order-ode}
\ddot{X}(\tau) + 2 X(\tau) + \gamma^{-1/2} V'(\gamma^{-1/2} X(\tau)) = 0.
\end{equation}
Thanks to assumption (P1),
the solution $X(\tau)$ is actually in $C^3_{\rm per}(0,T_0)$.
\end{proof}

The potential term of the nonlinear equation (\ref{second-order-ode})
is a singular contribution to the linear equation
as $\gamma \to 0$. Because of the singular contribution, $\| \ddot{X} \|_{L^{\infty}}$
grows as $\gamma \to 0$. Nevertheless, thanks to assumption (P4) of the compact support
of $V'(x)$, the solution $X(\tau)$ stays in the domain $|X| \geq a_0 \gamma^{1/2}$,
where $V'(\gamma^{-1/2} X) = 0$ for most of the times $\tau$ in the period $[0,T_0]$.
The following lemma estimates the size of the time interval, for which
the solution stays in the domain $|X| \leq a_0 \gamma^{1/2}$.

\begin{lemma}
\label{estimdeltat0}
Let $X(\tau)$ be the same as in Corollary \ref{lemma-H2-solution}.
Let $\Delta T_0$ be the measure of the subset of
$[0,T_0]$ in which $|X(\tau)| \leq a_0 \gamma^{1/2}$.
Then, $\Delta T_0$ admits the asymptotic expansion
\begin{equation}
\label{bounds-on-Delta-T} \Delta T_0 = 2 \sqrt{2} \gamma^{1/2} \int_0^{a_0} \frac{dx}{(E - V(x))^{1/2}} +
{\cal O}(\gamma^{3/2}) \quad \mbox{\rm as} \quad \gamma \to 0.
\end{equation}
\label{lemma-bounds-on-Delta-T}
\end{lemma}

\begin{proof}
Consider the splitting of $[0,T_0]$ into
\begin{eqnarray}
\nonumber
\left[ 0,\frac{1}{4}(T_0 - \Delta T_0) \right] \cup
\left[ \frac{1}{4}(T_0 - \Delta T_0),\frac{1}{4}(T_0 + \Delta T_0) \right] \cup
\left[ \frac{1}{4}(T_0 + \Delta T_0),\frac{1}{4}(3 T_0 - \Delta T_0) \right] \\
\phantom{textttext} \cup \left[ \frac{1}{4}(3 T_0 - \Delta T_0),\frac{1}{4}(3 T_0 + \Delta T_0) \right] \cup
\left[ \frac{1}{4}(3 T_0 + \Delta T_0),T_0 \right]. \label{splitting-interval}
\end{eqnarray}
Thanks to the symmetries (\ref{symmetry}), we have
$$
X\left(\frac{1}{4}T_0\right) = X\left(\frac{3}{4}T_0\right) = 0.
$$
In the first, third, and fifth intervals, the second-order equation (\ref{second-order-ode})
for sufficiently small $\gamma > 0$ becomes the linear oscillator
$$
\ddot{X} + 2 X = 0.
$$
An explicit solution with $X(0) = (E - V_{\infty})^{1/2}$ and $\dot{X}(0) = 0$
has the form
\begin{equation}
\label{unperturbed-solution} X(\tau) = (E - V_{\infty})^{1/2}
\cos(\sqrt{2}\tau), \quad \tau \in \left[ -\frac{1}{4}(T_0 -
\Delta T_0),\frac{1}{4}(T_0 - \Delta T_0)\right].
\end{equation}
The matching condition $X(\tau_0) = a_0 \gamma^{1/2}$ at $\tau_0 = \frac{1}{4} (T_0 - \Delta T_0)$ gives
$$
\cos\left( \frac{T_0 - \Delta T_0}{2 \sqrt{2}} \right) = \frac{a_0 \gamma^{1/2}}{(E - V_{\infty})^{1/2}}.
$$
Using the asymptotic expansion (\ref{period-expansion-T0}) for $T_0$, we obtain
the asymptotic expansion (\ref{bounds-on-Delta-T}) for $\Delta T_0$.
\end{proof}

\begin{corollary}
Let $X(\tau)$ be the same as in Lemma \ref{lemma-bounds-on-Delta-T}
and
\begin{equation}
\label{variable-Y}
Y(\tau) := \gamma^{-1/2} V'(\gamma^{-1/2} X(\tau)).
\end{equation}
Then, $Y \in C^1_{\rm per}(0,T_0)$ and
there exists $C(E) > 0$ such that $\| Y \|_{H^1_{\rm per}} \leq C
\gamma^{-3/4}$. \label{lemma-H4-solution}
\end{corollary}

\begin{proof}
Using the bounds on $\Delta T_0$ in Lemma \ref{lemma-bounds-on-Delta-T},
Corollary \ref{lemma-H2-solution}, and assumption (P1) on the potential $V(x)$,
we obtain
\begin{eqnarray*}
\int_0^{T_0} |Y(\tau)|^2 d \tau & = & \gamma^{-1} \int_0^{T_0} |V'(\gamma^{-1/2} X(\tau))|^2 d\tau \leq
\gamma^{-1} \Delta T_0 \| V' \|_{L^{\infty}}^2 \leq C_1 \gamma^{-1/2},\\
\int_0^{T_0} |\dot{Y}(\tau)|^2 d \tau & = & \gamma^{-2} \int_0^{T_0} |\dot{X}(\tau)|^2
|V''(\gamma^{-1/2} X(\tau))|^2 d\tau \leq
\gamma^{-2} \Delta T_0 \| V'' \|_{L^{\infty}}^2 \| \dot{X} \|^2_{L^{\infty}} \leq C_2 \gamma^{-3/2},
\end{eqnarray*}
for some constants $C_1,C_2 > 0$. The bound on $\| Y \|_{H^1_{\rm per}}$ follows from
the above computation.
\end{proof}

We shall be working in the space of functions in $H^2_{\rm per}(0,T_0)$, $H^1_{\rm per}(0,T_0)$, and
$L^2_{\rm per}(0,T_0)$ satisfying symmetry (\ref{symmetry}). Therefore, let us denote for all $p\geq 1$
\begin{equation}
\label{spaces}
H^p_e := \left\{ X \in H^p_{\rm per}(0,T_0) : \; X(-\tau) = X(\tau) = -X\left( \frac{T_0}{2} - \tau \right),
\quad \tau \in \R \right\}
\end{equation}
and use similar notations for $L^2_e$ and $L^{\infty}_e$.

We are now prepared to deal with the singularly perturbed linear oscillator under the
small source term:
\begin{equation}
\label{source-perturbation} \ddot{Z}(\tau) + 2 Z(\tau) +
\gamma^{-1/2} V'(\gamma^{-1/2} Z(\tau)) = \gamma^{\varepsilon+1/2} F(\tau),
\end{equation}
where $\varepsilon > 0$, $F \in L^2_e$, and $\| F \|_{L^2_{\rm per}} = {\cal O}(1)$ as $\gamma \to 0$.
It is necessary to consider the inhomogeneous problem (\ref{source-perturbation})
in order to control the effect of small coupling in the discrete Klein--Gordon equation (\ref{KGlattice})
at the central site $n = 0$. Energy for the perturbed oscillator equation
(\ref{source-perturbation}) can be written in the form
\begin{equation}
\label{energy-1} H(\tau) = \frac{1}{2} \dot{Z}^2 + Z^2 +
V(\gamma^{-1/2} Z).
\end{equation}

Because the homogeneous equation with $F \equiv 0$ admits a $T_0$-periodic
solution $X \in H^2_e$ with $\partial_E T_0(E,\gamma) = {\cal O}(\gamma^{1/2})$
(equation (\ref{source-perturbation}) is a singular perturbation of a
linear oscillator),
a source term
of order one would generate
a large output as $\gamma \to 0$, i.e.
the output $\| Z - X \|_{H^1_{\rm per}}$ is
going to be $\gamma^{-1/2}$ larger than the source term, roughly speaking.
This can be intuitively understood by linearizing equation (\ref{source-perturbation}) around $X$.
Indeed, the linearized operator
\begin{equation}
\label{operator-L-tau} L_0 := \frac{d^2}{d \tau^2} + 2 +
\gamma^{-1} V''\left(\gamma^{-1/2} X(\tau) \right)
\end{equation}
admits a nontrivial kernel ${\rm Ker}(L_0) = {\rm span}\{ \dot{X}
\}$ in $H^2_{\rm per}(0,T_0)$ and ${\rm Ker}(L_0) = \{0 \}$
in the subspace of even $T_0$-periodic functions, under the condition
$\partial_E T_0(E,\gamma) \neq 0$. In this subspace, we have
$\| L_0^{-1}   \|_{\mathcal{L}(L^2_{\rm per}, H^1_{\rm per} )}={\cal O}(\gamma^{-1/2})$,
due to the fact that
$\partial_E T_0(E,\gamma)= {\cal O}(\gamma^{1/2})$
(this follows from standard computations based on the variation of constants method).
Thanks to the scaling of the source term
considered in (\ref{source-perturbation}), we will get
$\| Z - X \|_{H^1_{\rm per}} = {\cal O}(\gamma^{\varepsilon})$ in Lemma \ref{lemma-spectrum},
if $\| F \|_{L^2_{\rm per}} = {\cal O}(1)$ and $\varepsilon \leq \frac{1}{4}$. Since $\| X
\|_{H^1_{\rm per}} = {\cal O}(1)$ as $\gamma \to 0$ by Corollary
\ref{lemma-H2-solution}, the perturbation $\| Z - X
\|_{H^1_{\rm per}}$ will be still smaller than the unperturbed solution $\| X
\|_{H^1_{\rm per}}$.

There exists an obstacle on the direct application of the Implicit Function
Theorem to obtain $T_0$-periodic solutions of the singularly perturbed oscillator equation (\ref{source-perturbation}).
The obstacle comes from the power series expansion
$$
V'(\gamma^{-1/2} Z) = V'(\gamma^{-1/2} X) + \gamma^{-1/2} V''(\gamma^{-1/2} X) (Z - X) +
\gamma^{-1} V'''(\gamma^{-1/2} X) (Z - X)^2 + ...,
$$
which generate large terms for $\gamma \to 0$ because of the singular perturbation in the
nonlinear potential. To avoid this difficulty, we use the fact that
$V'(\gamma^{-1/2}Z)$ has a compact support,
and transform the search of periodic solutions to
a root finding problem to which the Implicit
Function Theorem can be applied.
We shall prove the following.

\begin{lemma}
\label{lemma-spectrum} Let $X \in H^2_e$ be the
solution of Lemma \ref{lemma-oscillations} in parametrization
(\ref{scaling-transformation}) for any fixed $E > V_L$. Let $B_{\delta}$ be a ball of
radius $\delta > 0$ in $L^2_e$ centered at $0$.
Fix $\varepsilon \in \left( 0, \frac{1}{2} \right)$. If
\begin{equation}
\label{condition-inversion} \int_0^{a_0} \frac{dx}{(E -
V(x))^{3/2}} - \frac{a_0}{(E - V_{\infty})^{3/2}} \neq 0,
\end{equation}
there exist $\gamma_0(\varepsilon,\delta,E) > 0$ and $\eta (\varepsilon,\delta,E) > 0$
such that the inhomogeneous equation (\ref{source-perturbation}) with $\gamma \in (0,\gamma_0)$
and $F \in B_{\delta}$ admits a unique solution $Z = \mathcal{G}_{\gamma ,\varepsilon} (F) \in H^2_e$
satisfying
$$
Z(\tau_0)=a_0 \gamma^{1/2}, \ \ \ |H(\tau_0)-E|<\eta
$$
for some $\tau_0 \in (0, \frac{T_0}{4})$.
Moreover, $Z$ is close to $X$ in $H^1_e$
with
\begin{equation}
\label{bound-L2-F}
\| Z - X \|_{H^1_{\rm per}} \leq
C_0(\varepsilon,\delta,E) \gamma^{\tilde{\varepsilon}}, \quad
\tilde{\varepsilon} = \min\left(\frac{1}{4},\varepsilon \right)
\end{equation}
and satisfies the estimate
\begin{equation}
\label{boundC1}
\| Z \|_{C^1_{\rm per} } \leq 3 \sqrt{E}.
\end{equation}
In addition, $Z(\tau ) \geq a_0 {\gamma}^{1/2}$ for $\tau \in [0,\tau_0]$,
$Z(\tau ) \in [0,a_0 {\gamma}^{1/2}]$ for $\tau \in [\tau_0, \frac{T_0}{4}]$,
and there exists $\theta(\varepsilon,\delta,E) >0$ such that
\begin{equation}
\label{deltat0}
\left| \frac{T_0}{4}-\tau_0 \right| \leq \theta(\varepsilon,\delta,E) \gamma^{\frac{1}{2}}.
\end{equation}
\end{lemma}

\begin{proof}
We shall use a kind of shooting method to transform the differential equation
(\ref{source-perturbation})
to a root-finding problem. Let us consider an initial-value problem for the second-order
differential equation (\ref{source-perturbation}) starting
with the initial data $Z(0) = Z_0$ and $\dot{Z}(0) = 0$, where $Z_0$ is a positive
$\gamma$-independent parameter.

Thanks to the compact support in assumption (P4), we
solve the inhomogeneous linear equation
\begin{equation}
\label{linear-oscillator-perturbed}
\ddot{Z}(\tau) + 2 Z(\tau) = \gamma^{\varepsilon +1/2} F(\tau), \quad \tau \in
[0,\tau_0],
\end{equation}
where $\tau_0$ will be determined by the condition $Z(\tau_0) = a_0 \gamma^{1/2}$.
We shall prove that for small $\gamma > 0$, a unique $\tau_0 \in \left( 0, \frac{\pi}{2\sqrt{2}} \right)$ exists.
In what follows we denote $W_0 := \dot{Z}(\tau_0)$ and $\tau_0 := \frac{\pi}{2 \sqrt{2}} - \Delta_0$.
We shall consider $W_0 \in \R_-$ as a free parameter and express $(Z_0 , \tau_0)$
(or equivalently $(Z_0 , \Delta_0)$) as a function of $W_0$.

The unique solution of the linear equation (\ref{linear-oscillator-perturbed})
with the initial data $Z(0) = Z_0$ and $\dot{Z}(0) = 0$ is given by
\begin{equation}
\label{solution-linear-oscillator-perturbed}
Z(\tau) = Z_0 \cos(\sqrt{2} \tau) + \frac{1}{\sqrt{2}} \gamma^{\varepsilon +1/2} \int_0^{\tau}
F(\tau') \sin(\sqrt{2}(\tau - \tau')) d \tau'.
\end{equation}
At $\tau = \tau_0$, we have the system of nonlinear equations
\begin{eqnarray*}
\left\{ \begin{array}{l} Z_0 \sin(\sqrt{2} \Delta_0) = a_0 \gamma^{1/2} - \frac{1}{\sqrt{2}}
\gamma^{\varepsilon +1/2} \int_0^{\tau_0}
F(\tau') \sin(\sqrt{2}(\tau_0 - \tau')) d \tau', \\
Z_0 \cos(\sqrt{2} \Delta_0) = -\frac{1}{\sqrt{2}} W_0 + \frac{1}{\sqrt{2}} \gamma^{\varepsilon +1/2} \int_0^{\tau_0} F(\tau') \cos(\sqrt{2}(\tau_0 - \tau')) d \tau', \end{array} \right.
\end{eqnarray*}
which can be rewritten in the equivalent form
\begin{eqnarray*}
\left\{ \begin{array}{l}
Z_0 = -\frac{1}{\sqrt{2}} W_0 \cos(\sqrt{2} \Delta_0) +
a_0 \gamma^{1/2} \sin(\sqrt{2} \Delta_0) + \frac{1}{\sqrt{2}} \gamma^{\varepsilon +1/2}
\int_0^{\tau_0} F(\tau) \sin(\sqrt{2} \tau) d \tau, \\
\frac{1}{\sqrt{2}} W_0  \sin(\sqrt{2} \Delta_0) = - a_0 \gamma^{1/2} \cos(\sqrt{2} \Delta_0)
- \frac{1}{\sqrt{2}} \gamma^{\varepsilon +1/2}
\int_0^{\tau_0} F(\tau) \cos(\sqrt{2} \tau) d \tau
 \end{array} \right.
\end{eqnarray*}
(multiply each equation by $\sin(\sqrt{2} \Delta_0)$ and $\cos(\sqrt{2} \Delta_0)$ and
sum the resulting equations). This problem can be rewritten
\begin{eqnarray}
\label{systzd}
\left\{ \begin{array}{l}
Z_0 = -\frac{1}{\sqrt{2}} W_0 \cos(\sqrt{2} \Delta_0) +
a_0 \gamma^{1/2} \sin(\sqrt{2} \Delta_0) + \frac{1}{\sqrt{2}} \gamma^{\varepsilon +1/2}
\int_0^{\frac{\pi}{2 \sqrt{2}} - \Delta_0} F(\tau) \sin(\sqrt{2} \tau) d \tau, \\
\Delta_0
=
- \frac{a_0}{W_0\, \mbox{sinc}(\sqrt{2}\Delta_0)} \gamma^{1/2} \cos(\sqrt{2} \Delta_0)
- \frac{1}{\sqrt{2} W_0 \, \mbox{sinc}(\sqrt{2}\Delta_0)} \gamma^{\varepsilon +1/2}
\int_0^{\frac{\pi}{2 \sqrt{2}} - \Delta_0} F(\tau) \cos(\sqrt{2} \tau) d \tau ,
 \end{array} \right.
\end{eqnarray}
where $\mbox{sinc(x)}=\sin{(x)}/x$.
Assuming that $W_0 \in (-\infty , W_{\rm{max}})$ is bounded away from $0$ and
$\gamma \in (0,\gamma_{\rm{max}})$ is small enough,
the second equation can be solved by the contraction mapping theorem
for $\Delta_0 \in \left[ 0, \frac{\pi}{2 \sqrt{2}} \right]$, and then the first equation
determines $Z_0$. This yields finally
\begin{equation}
\label{equation-1}
Z_0 = -\frac{1}{\sqrt{2}} W_0 + \frac{1}{\sqrt{2}} \gamma^{\varepsilon +1/2}
\int_0^{\frac{\pi}{2 \sqrt{2}}} F(\tau) \sin(\sqrt{2} \tau) d \tau + {\cal O}(\gamma)
\end{equation}
and
\begin{equation}
\label{equation-2}
\Delta_0 = -\frac{a_0 \gamma^{1/2}}{W_0} + \frac{1}{\sqrt{2} W_0} \gamma^{\varepsilon +1/2}
\int_0^{\frac{\pi}{2 \sqrt{2}}} F(\tau) \cos(\sqrt{2} \tau) d \tau + {\cal O}(\gamma^{3/2}).
\end{equation}
Then $Z_0 = {\cal O}(1)$ and $\Delta_0 = {\cal O}(\gamma^{1/2})$
as $\gamma \to 0$, so that existence of a unique $\tau_0 \in \left( 0, \frac{\pi}{2\sqrt{2}} \right)$ follows.

We can now continue solution $Z(\tau)$ to $\tau > \tau_0$ starting with the initial conditions $Z(\tau_0) = a_0 \gamma^{1/2}$ and $\dot{Z}(\tau_0) = W_0 < 0$. Our aim is to solve the inhomogeneous differential equation
\begin{equation}
\label{nonlinear-oscillator-perturbed}
\ddot{Z}(\tau) + 2 Z(\tau) + \gamma^{-1/2} V'(\gamma^{-1/2} Z(\tau)) = \gamma^{\varepsilon +1/2} F(\tau),
\quad \tau \in [\tau_0,\tau_*],
\end{equation}
where $\tau_* > \tau_0$ is the first time where $Z(\tau_*) = 0$.
For this purpose the first step is
to show that $\tau_* > \tau_0$ exists.

Taking the derivative of the energy $H$ in (\ref{energy-1}) with respect to $\tau$ and using equation (\ref{nonlinear-oscillator-perturbed}),
we infer that
\begin{equation}
\label{energy-1-der}
\dot{H} = \gamma^{\varepsilon +1/2} \dot{Z} F.
\end{equation}
Let $E^\ast := H(\tau_0 )=\frac{1}{2} W_0^2 + a_0^2 \gamma + V_{\infty}$
and assume that $E^\ast > V_L \geq V_{\infty}$. Note that $E^\ast$ becomes now
the parameter of the solution family in place of $W_0$.

Let us denote by $\tau_1$ the maximal time such that
$|Z(\tau)| < a_0 \gamma^{1/2}$ and $\dot{Z}(\tau)<0$
for all $\tau \in (\tau_0,\tau_1 )$.
We have
\begin{equation}
\label{derivative-formula}
\dot{Z}(\tau) = -\sqrt{2 (H(\tau) - V(\gamma^{-1/2} Z(\tau)) - Z^2(\tau))} , \quad \tau \in [\tau_0,\tau_1] .
\end{equation}
Using formula (\ref{derivative-formula}) and integrating
equation (\ref{energy-1-der}), we obtain for all $\tau  \in [\tau_0,\tau_1] $
\begin{eqnarray*}
|H(\tau) - E^\ast | \leq
\gamma^{\varepsilon +1/2} \int_{\tau_0}^{\tau} \sqrt{2 (H(\tau') - E^\ast + E^\ast - V(\gamma^{-1/2} Z(\tau'))
- Z^2(\tau'))} |F(\tau')| d \tau'.
\end{eqnarray*}
Using the triangle inequality, we find that there exist $\gamma_0 > 0$ and a $\gamma$-independent constant
$C(E^\ast,\varepsilon ) > 0$ such that for all $\gamma \in (0,\gamma_0)$,
\begin{eqnarray}
\label{control-H}
|H(\tau) - E^\ast | \leq C(E^\ast,\varepsilon ) \gamma^{\varepsilon +1/2}  \| F \|_{L^2_{\rm per}}, \quad \tau \in [\tau_0,\tau_1].
\end{eqnarray}
Now it follows that
$$
 \frac{1}{2} \dot{Z}^2 (\tau_1)=
H(\tau_1) - Z^2 (\tau_1) - V(\gamma^{-1/2} Z(\tau_1))
\geq E^\ast -V_L + {\cal O}(\gamma^{\varepsilon +1/2}),
$$
hence $\dot{Z}(\tau_1)\neq 0$ for $\gamma$ small enough since $ E^\ast >V_L$.
Consequently one has $Z(\tau_1) = -a_0 \gamma^{1/2}$ and
$\dot{Z}(\tau)<0$
for all $\tau \in [\tau_0,\tau_1 ]$.
By the intermediate value theorem,
this yields the existence of $\tau_* \in (\tau_0,\tau_1 )$ such that
$Z(\tau_*) = 0$ and $Z(\tau ) \in (0, a_0 \gamma^{1/2})$ for all
$\tau \in (\tau_0,\tau_* )$.

We now express the distance $|\tau_* - \tau_0|$ from the energy (\ref{energy-1})
controlled by bound (\ref{control-H}). From equations (\ref{energy-1-der}) and (\ref{derivative-formula}),
we obtain
\begin{eqnarray*}
\tau_* - \tau_0 = \int_{\tau_0}^{\tau_*} \frac{\dot{Z}(\tau)}{\dot{Z}(\tau)} d \tau =
-\int_{\tau_0}^{\tau_*} \frac{\dot{Z}(\tau) d \tau}{\sqrt{2 (H(\tau) - V(\gamma^{-1/2} Z(\tau)) - Z^2(\tau))}}.
\end{eqnarray*}
Using bound (\ref{control-H}), we have for all $\tau \in [\tau_0,\tau_*]$ as $\gamma \to 0$,
$$
\frac{1}{\sqrt{2 (H(\tau) - V(\gamma^{-1/2} Z(\tau)) - Z^2(\tau))}} =
\frac{1}{\sqrt{2 (E^\ast - V(\gamma^{-1/2} Z(\tau)) - Z^2(\tau))}} + {\cal O}(\gamma^{\varepsilon +1/2}),
$$
so that the change of variables $x = \gamma^{-1/2} Z(\tau)$ gives
\begin{eqnarray}
\label{equation-5}
\tau_* - \tau_0 = \gamma^{1/2} \int_{0}^{a_0} \frac{dx}{\sqrt{2 (E^\ast - V(x) - \gamma x^2)}}
+ {\cal O}(\gamma^{\varepsilon +1/2}) \quad \mbox{\rm as} \quad \gamma \to 0.
\end{eqnarray}
Combining this expansion of $\tau_* - \tau_0$ with the expansion of $\tau_0 = \frac{\pi}{2 \sqrt{2}} - \Delta_0$
obtained in (\ref{equation-2}), and using the fact that
$   W_0  =-\sqrt{2(E^\ast -  V_{\infty} - a_0^2 \gamma )} $, we end up with the expansion
\begin{eqnarray}
\label{root-period}
\tau_* (E^\ast, \gamma , F )=  \frac{\pi}{2 \sqrt{2}} - \frac{a_0 \gamma^{1/2}}{\sqrt{2(E^\ast - V_{\infty})}}
+ \gamma^{1/2} \int_{0}^{a_0} \frac{dx}{\sqrt{2 (E^\ast - V(x))}} + {\cal O}(\gamma^{\varepsilon +1/2})
\quad \mbox{\rm as} \quad \gamma \to 0,
\end{eqnarray}
uniformly in $F \in B_{\delta}$ for
fixed $\varepsilon >0$ and $\delta > 0$.

Let us now examine the regularity of $\tau_*$ with respect to the variable $(E^\ast, \gamma , F )$ including
the functional parameter $F$. From the standard fixed-point reformulation of the Cauchy problem for differential equations, it follows that the map $(W_0 , \gamma , F ) \mapsto Z$ is $C^1$ from
$(-\infty , W_{\rm{max}})\times (0,\gamma_{\rm{max}}) \times L^2_e$ into $C^1([0,\frac{\pi}{\sqrt{2}}])$.
Consequently, the map $(E^\ast , \gamma , F , \tau ) \mapsto Z(\tau )$ is $C^1$ from
$(E_{\rm{min}},+\infty  )\times (0,\gamma_{\rm{max}}) \times L^2_e \times \left(0,\frac{\pi}{\sqrt{2}}\right)$ into
$\R$, where we have denoted $E_{\rm{min}}=\frac{1}{2} W_{\rm{max}}^2 + V_{\infty}$.
Since ${Z}(\tau_\ast) = 0$ and $\dot{Z}(\tau_\ast) \neq 0$, the implicit function theorem ensures that
$\tau_\ast$ is $C^1$ with respect to $(E^\ast, \gamma , F ) \in
(E_{\rm{min}},+\infty  )\times (0,\gamma_{\rm{max}}) \times B_\delta$.

Now we come back to equation (\ref{source-perturbation}).
If we manage to find a value of $E^\ast$ such that
\begin{equation}
\label{rootf}
\tau_* (E^\ast, \gamma , F )=\frac{1}{4}\, T_0(E,\gamma ),
\end{equation}
then the symmetry constraints satisfied by $F \in L^2_e$ in (\ref{spaces})
imply that $Z \in H^2(0,T_0)$ also satisfies the symmetry conditions in (\ref{spaces}).
Indeed, $Z(\tau)$ and $Z(-\tau)$ are solutions of the same Cauchy problem, hence
$Z(\tau)=Z(-\tau)$ by Cauchy's theorem. Using the same argument
one obtains $Z(\tau + T_0 /4)=-Z(-\tau + T_0 /4)$ (this requires in addition the evenness of $V$).
These two equalities imply $Z(0)=-Z(T_0 /2)=-Z(-T_0 /2)=Z(T_0)$ and
in the same way $\dot{Z}(0)=\dot{Z}(T_0)=0$. Consequently,
$Z(\tau)$ and $Z(\tau + T_0)$ are solutions of the same Cauchy problem, hence
$Z$ is $T_0$-periodic. Therefore, $Z \in H^2_e$ if there is $E^*$
solving equation (\ref{rootf}).

Using expansions (\ref{period-expansion-T0}) and (\ref{root-period}), equation (\ref{rootf}) can be rewritten
\begin{equation}
\label{rootexp}
R(E^\ast , \gamma^{ \varepsilon} )=0,
\end{equation}
where
\begin{equation}
\label{rootexpression}
R(E^\ast,\gamma^{ \varepsilon})=
\int_{0}^{a_0} \frac{dx}{\sqrt{ E - V(x)}} - \frac{a_0 }{\sqrt{E - V_{\infty}}}
 -\int_{0}^{a_0} \frac{dx}{\sqrt{ E^\ast - V(x)}} - \frac{a_0 }{\sqrt{E^\ast - V_{\infty}}}
 + {\cal O}(\gamma^{ \varepsilon})
\end{equation}
as $\gamma \rightarrow 0$, uniformly in $F \in B_\delta$ (we omit the dependency of $R$ with respect to $E$ and $F$ in notations). Now let us consider the non-degeneracy condition
(\ref{condition-inversion}), which is equivalent to $\frac{\partial R}{\partial E^\ast }(E,0) \neq 0$.
Due to the non-differentiability of $R$ at $\gamma =0$, we cannot solve (\ref{rootexp}) directly
by the implicit function theorem, but rewrite the problem in the form
\begin{equation}
\label{rootexp2}
E^\ast - E = {\cal O}(\gamma^{ \varepsilon} + (E^\ast - E)^2)
\end{equation}
using the fact that $\frac{\partial R}{\partial E^\ast }(E,0) \neq 0$.
By the contraction mapping theorem,
if $E>V_L$, $\varepsilon \in \left(0,\frac{1}{2}\right)$, $\delta > 0$, $F \in B_\delta$ are
fixed, there exist constants $\gamma_0(\varepsilon,\delta,E) > 0$ and $C_0(\varepsilon,\delta,E)$ such that equation (\ref{rootexp2}) admits a unique solution $E^\ast$ near $E$ for all $\gamma \in (0,\gamma_0)$, with
\begin{equation}
\label{bound-energy}
|E^\ast - E| \leq C(\varepsilon,\delta, E) \gamma^{\varepsilon} .
\end{equation}
Moreover, the map $(\gamma , F) \mapsto E^\ast$ is $C^1$ on $(0,\gamma_0) \times B_\delta$
by the uniform contraction principle.

This completes the proof of existence and local uniqueness of solution $Z$. Estimate (\ref{deltat0})
follows directly from expansion (\ref{equation-5}).

It remains to prove bounds (\ref{bound-L2-F}) and (\ref{boundC1}).
For this purpose,
we first consider a time interval in which $Z$ and $X$
are both given by the explicit solution of a linear equation.
Recall that
\begin{equation}
\label{solution-linear-oscillator-perturbed-again}
X(\tau )=\sqrt{E-V_\infty}\, \cos{(\sqrt{2}\, \tau)}, \quad
0\leq \tau \leq \tilde{\tau}_0,
\end{equation}
where
$$
\tilde{\tau}_0 = \frac{1}{\sqrt{2}}\, \mbox{arccos}
\Big( \frac{a_0 \gamma^{1/2}}{(E - V_{\infty})^{1/2}} \Big) .
$$
Now consider $\hat{\tau}_0 = \min(\tau_0 , \tilde{\tau}_0 )=\frac{T_0}{4}+ {\cal O}(\gamma^{1/2})$.
By combining the explicit expressions of $X$ and $Z$
(see equation (\ref{solution-linear-oscillator-perturbed})),
estimates (\ref{equation-1}), (\ref{bound-energy})
and observing that
$$
W_0 = -\sqrt{2(E^\ast -V_\infty)}+ {\cal O}(\gamma),
$$
one obtains
\begin{equation}
\label{bound-l-infty}
\| Z - X \|_{L^{\infty}(0,\hat{\tau}_0)} \leq
C_0(\varepsilon,\delta,E) \gamma^{\varepsilon}, \quad
\| \dot{Z} - \dot{X} \|_{L^{\infty}(0,\hat{\tau}_0)} \leq
C_0(\varepsilon,\delta,E) \gamma^{\varepsilon} .
\end{equation}

Since $\hat{\tau}_0 = \frac{T_0}{4}+ {\cal O}(\gamma^{1/2})$, it follows
from equation (\ref{solution-linear-oscillator-perturbed}) and (\ref{solution-linear-oscillator-perturbed-again})
that
\begin{equation}
\label{infinity-piece}
\| X \|_{L^{\infty}(\hat{\tau}_0 , T_0 /4 )} + \| Z \|_{L^{\infty}(\hat{\tau}_0 , T_0 /4 )} =
{\cal O}(\gamma^{1/2}) \quad \mbox{\rm as} \quad \gamma \to 0.
\end{equation}
As a result, we extend the first bound in (\ref{bound-l-infty}) to
$$
\| Z - X \|_{L^{\infty}} \leq C_0(\varepsilon,\delta,E) \gamma^{\varepsilon}.
$$
To extend the second bound in (\ref{bound-l-infty}), we write for all $\tau \in \left[ \hat{\tau}_0, \frac{1}{4} T_0 \right]$,
\begin{eqnarray}
\label{estimc1}
\left\{ \begin{array}{l}
|\dot{Z}(\tau)| = \sqrt{2(H(\tau) - V(\gamma^{-1/2} Z(\tau)) - Z^2(\tau))} \leq \sqrt{2H(\tau )} = \sqrt{2 E}+ {\cal O}(\gamma^{\varepsilon}), \\
|\dot{X}(\tau)| = \sqrt{2(E - V(\gamma^{-1/2} X(\tau)) - X^2(\tau))}  \leq \sqrt{2E}. \end{array} \right.
\end{eqnarray}
where we have used bounds (\ref{control-H}) and (\ref{bound-energy}). By bounds (\ref{equation-5})
and (\ref{infinity-piece}), hence we have
\begin{equation}
\label{bound-l-infty-bis}
\| X \|_{H^{1}(\hat{\tau}_0 , T_0 /4 )} + \| Z \|_{H^{1}(\hat{\tau}_0 , T_0 /4 )} \leq C_0(\varepsilon,\delta,E) \gamma^{1/4}.
\end{equation}
As a conclusion, estimate (\ref{bound-L2-F}) follows from estimates
(\ref{bound-l-infty}) and (\ref{bound-l-infty-bis}).
Estimate (\ref{boundC1}) follows from (\ref{bound-l-infty}) and (\ref{estimc1})
for $\gamma$ small enough.
\end{proof}

\begin{remark}
A $T_0$-periodic solution of the perturbed oscillator equation (\ref{source-perturbation})
may not exist if the constraint (\ref{condition-inversion}) is not met. For instance, if $V \equiv 0$
(this corresponds to the case $a_0 = 0$)
and $F(\tau) = \cos(\sqrt{2} \tau)$, the perturbed linear oscillator
(\ref{source-perturbation}) has no $T_0$-periodic solutions
because of the resonance of the linear oscillator and the
$T_0$-periodic force for $T_0 = \sqrt{2}\pi$.
Since $\Delta_0 = 0$,
expansion (\ref{equation-2}) yields the
classical constraint
$$
\int_0^{\frac{\pi}{2\sqrt{2}}} F(\tau) \cos(\sqrt{2} \tau) d \tau =0
$$
for the existence of a $\sqrt{2}\pi$-periodic solution
under symmetric periodic forcing. The constraint is not satisfied if $F(\tau) = \cos(\sqrt{2} \tau)$.
\end{remark}

For the arguments in the proof of Theorem \ref{theorem-main} based on the Schauder fixed-point
theorem, we will need a continuity of the nonlinear map $F \mapsto Z$
from $L^2_e$ to $H^2_e$.

\begin{lemma}
\label{lemma-continuity} Under Assumptions of Lemma \ref{lemma-spectrum},
for all $\gamma \in (0,\gamma_0)$ and all $F \in B_{\delta}$,
the solution $Z=\mathcal{G}_{\gamma , \epsilon}(F) \in H^2_e$ is continuous
with respect to $F \in L^2_e$.
\end{lemma}

\begin{proof}
From Lemma \ref{lemma-spectrum}, we know that
for given $F_1 \in B_{\delta}$ and $F_2 \in B_{\delta}$,
there exist roots $E^* = E_1$ and $E^* = E_2$ of equation (\ref{rootexp}) and
if $F_1 \to F_2$ in $L^2_e$, then $E_1 \to E_2$.
For each $E_1$ and $E_2$, we adapt the notations of Lemma \ref{lemma-spectrum} by using the symbols
$W_{1,2} = -\sqrt{2(E_{1,2} - V_{\infty} - a_0^2 \gamma)}$, $Z_0^{(1),(2)}$ and $\tau_{1,2}$
instead of $W_0$, $Z_0$ and $\tau_0 = \frac{\pi}{2 \sqrt{2}} - \Delta_0$ respectively.
Expansion (\ref{equation-2}) implies that
$$
\exists C > 0 : \quad |\tau_1 - \tau_2| \leq C
\gamma^{1/2} \left( |E_1 - E_2| + \gamma^{\varepsilon} \| F_1 - F_2 \|_{L^2(0,T_0)} \right).
$$
Let $\hat{\tau}_0 = \min(\tau_1,\tau_2)$. From the solution (\ref{solution-linear-oscillator-perturbed})
of the linear oscillator equation (\ref{linear-oscillator-perturbed}) and the expansion
(\ref{equation-1}), we obtain
\begin{equation}
\label{estimlin}
\exists C > 0 : \quad \|Z_1 - Z_2 \|_{H^2_{\rm per}(0,\hat{\tau}_0)}
\leq C \left( |E_1 - E_2| + \gamma^{\varepsilon + 1/2} \| F_1 - F_2 \|_{L^2(0,T_0)} \right).
\end{equation}
Consequently, if $F_1 \to F_2$ in $L^2_e$, then $E_1 \to E_2$ and $Z_1 \to Z_2$ in $H^2_{\rm per}(0,\hat{\tau}_0)$.

There remains to complete (\ref{estimlin}) by a continuity result on
the time interval $[\hat{\tau}_0,\tau_*]$, on which we rewrite the differential equation (\ref{nonlinear-oscillator-perturbed}) in variables
\begin{equation}
\label{scaling-inner}
z_{1,2}(t) = \gamma^{-1/2} Z_{1,2}(\tau), \quad f_{1,2}(t) = F_{1,2}(\tau), \quad
t = \gamma^{-1/2} (\tau - \hat{\tau}_0),
\end{equation}
or explicitly,
\begin{equation}
\label{ivp-inner}
\ddot{z}_{1,2}(t) + 2 \gamma z_{1,2}(t) + V'(z_{1,2}(t)) = \gamma^{\varepsilon + 1} f_{1,2}(t), \quad t \in [0,\hat{t}_0],
\end{equation}
where $\hat{t}_0 = \gamma^{-1/2} (\tau_* - \hat{\tau}_0) = {\cal O}(1)$ thanks to (\ref{deltat0}) and
$\tau_* = T_0 /4$.

For definiteness, let us assume that $\tau_1 > \tau_2$ so that $\hat{\tau}_0 = \tau_2$.
The initial-value problems for differential equations (\ref{ivp-inner}) are started with
$(z_1(0),\dot{z}_1(0))$ and $(z_2(0),\dot{z}_2(0))$, where $z_2(0) = a_0$,
$\dot{z}_2(0) = \dot{Z}_2(\tau_2) = W_2$, and $z_{1}(0) = \gamma^{-1/2} Z_{1}(\tau_2)$,
$\dot{z}_1(0) = \dot{Z}_1(\tau_2)$ are determined by (\ref{solution-linear-oscillator-perturbed}).

Since $V'$ is globally Lipschitzian by the assumptions (P1) and (P4), Gronwall's inequality implies
the existence of $C>0$ such that
\begin{eqnarray*}
\quad \| z_1 - z_2 \|_{C^1(0,\hat{t}_0)} & \leq & C \left(
|z_1(0) - z_2(0)| + |\dot{z}_1(0) - \dot{z}_2(0)| + \gamma^{\epsilon + 1} \int_0^{\hat{t}_0} | f_1(t) - f_2(t)| dt \right) \\
& \leq & C \left(
 \frac{|Z_{1}(\tau_2)-Z_{2}(\tau_2)|}{\gamma^{1/2}} + |\dot{Z}_1(\tau_2) - \dot{Z}_2(\tau_2)| + \gamma^{\epsilon + 1}  \hat{t}_0^{1/2} \| f_1 - f_2 \|_{L^2(0,\hat{t}_0)} \right).
\end{eqnarray*}
When $F_1 \to F_2$ in $L^2_e$, the first equation of (\ref{systzd}) yields $Z_0^{(1)} \rightarrow Z_0^{(2)}$
(since $W_1 \rightarrow W_2$ and $\tau_1 \rightarrow \tau_2$),
which implies $Z_{1}(\tau_2) \rightarrow Z_{2}(\tau_2)$ and $\dot{Z}_1(\tau_2) \rightarrow \dot{Z}_2(\tau_2)$
thanks to expression (\ref{solution-linear-oscillator-perturbed}).
Consequently we obtain $\| z_1 - z_2 \|_{C^1(0,\hat{t}_0)} \rightarrow 0$ and
$\| Z_1 - Z_2 \|_{C^1(\hat{\tau}_0,\tau_*)} \rightarrow 0$.
Combining this result with (\ref{estimlin}), we see that if $F_1 \to F_2$ in $L^2_e$
then $Z_1 \to Z_2$ in $H^1_e$.

Now observing that $\Delta = \frac{d^2}{d\tau^2}\, : H^2_e \rightarrow L^2_e$ is invertible
and considering equation (\ref{source-perturbation}) that defines $\mathcal{G}_{\gamma , \epsilon}(F)$
implicitly, $Z_{1,2} = \mathcal{G}_{\gamma , \epsilon}(F_{1,2})$ satisfy the equality
$$
Z_{1,2} = \Delta^{-1}\, [\,  \gamma^{\varepsilon+1/2} F_{1,2} -2 Z_{1,2}
-  \gamma^{-1/2} V'(\gamma^{-1/2}  Z_{1,2})  \, ].
$$
Hence continuity in $H^1_e$ norm implies continuity in $H^2_e$ norm.
\end{proof}

\section{Small-amplitude oscillations on other sites}
\label{section-3}

In our construction, the large-amplitude breather bifurcating from
infinity as $\gamma \to 0$ is localized at a
single site $n = 0$, and close to the periodic solution of
Lemma \ref{lemma-oscillations}. The other amplitudes for $n \neq
0$ display the oscillatory motion guided by the central
oscillation at $n = 0$ and powered by a small amplitude in
$\gamma$. By symmetry of the discrete Laplacian, we may assume
that
$$
x_n = x_{-n}, \quad n \geq 1.
$$
Since the amplitudes of
oscillations for $n \neq 0$ are small, we shall consider the
linearized discrete Klein--Gordon equation (\ref{KGlattice}) at the zero
solution to study possible resonances with the oscillatory motion
at $n = 0$.

To this end, let us introduce the function spaces
$$
\mathbb{X} := L^2_{\rm per}((0,T_0);l^2(\mathbb{N})), \quad
\mathbb{D} := H^2_{\rm per}((0,T_0);l^2(\mathbb{N})),
$$
and denote $\omega_0 = \frac{2\pi}{T_0}$.
Let us use assumption (P2) and consider
the inhomogeneous linear problem
\begin{equation}
\label{inhomogen-problem}
\begin{array}{l}
\gamma U''_1(\tau) + \kappa^2 U_1(\tau) = \gamma
(U_{2}(\tau)  - 2 U_1(\tau)) + F_1(\tau), \\
\\
\gamma U''_n(\tau) + \kappa^2 U_n(\tau) = \gamma
(U_{n+1}(\tau) + U_{n-1}(\tau) - 2 U_n(\tau)) + F_n(\tau), \quad n
\geq 2,
\end{array}
\end{equation}
where $\tau = \gamma^{1/2} t$ denotes the rescaled time.
We denote the sequence
$\{ F_n(\tau ) \}_{n \geq 1} \in \mathbb{X}$ by ${\bf F}$  and look for
solution ${\bf U}=\{ U_n(\tau ) \}_{n \geq 1} \in \mathbb{D}$  of (\ref{inhomogen-problem}).
We shall prove the following.

\begin{lemma}
Assume the non-resonance condition (\ref{resonances}) to be satisfied
for given values of $\kappa$, $\gamma $ and $\omega_0$, i.e.
\begin{equation}
\label{resonances-again} \kappa^2 + 2 \gamma (1 -
\cos{q}) \neq  m^2 \omega_0^2 \gamma , \quad
\forall\, m \in \Z, \quad \forall q \in [-\pi,\pi].
\end{equation}
Then, for all ${\bf F}\in \mathbb{X}$ problem (\ref{inhomogen-problem})
admits a unique solution ${\bf U} \in \mathbb{D}$.
Moreover, one has
\begin{equation}
\label{bound-resolvent-intermediate}
\| {\bf U} \|_{\mathbb{X} } \leq
C_R (\gamma , \omega_0 , \kappa ) \| {\bf F} \|_{\mathbb{X}},
\end{equation}
where
\begin{equation}
\label{constant}
C_R (\gamma , \omega_0 , \kappa )=
\Big(
{\displaystyle \inf_{m\in \mathbb{Z}, \, q \in [-\pi,\pi]} \left| \kappa^2 - m^2
\omega_0^2 \gamma + 2 \gamma (1 - \cos{q}) \right| }
\Big)^{-1} < +\infty .
\end{equation}
Moreover, if in addition ${\bf F} \in H^1_{\rm per}((0,T_0);l^2(\mathbb{N}))$ then
\begin{equation}
\label{bound-resolvent-intermediate-2}
\| {\bf U} \|_{ H^1_{\rm per}((0,T_0);l^2(\mathbb{N}))  } \leq
C_R (\gamma , \omega_0 , \kappa ) \| {\bf F} \|_{ H^1_{\rm per}((0,T_0);l^2(\mathbb{N}))  }.
\end{equation}
\label{lemma-4}
\end{lemma}

\begin{proof}
In order to solve (\ref{inhomogen-problem}), we expand ${\bf F}$ and ${\bf U}$
using Fourier series in the time variable $t$ and band-limited Fourier transforms in the discrete
spatial coordinate $n$. Defining
$$
\hat{f}_m(q)=
\frac{1}{\pi T_0 }\sum_{n\geq 1}\int_{0}^{T_0}{F_n(\tau ) \, \sin(n q)\, e^{-i m \omega_0 \tau} \, d\tau},
$$
$$
\hat{u}_m(q)=
\frac{1}{\pi T_0 }\sum_{n\geq 1}\int_{0}^{T_0}{U_n(\tau ) \, \sin(n q)\, e^{-i m \omega_0 \tau}\, d\tau},
$$
we have thus
$$
F_n(\tau) = \sum_{m \in \Z} \int_{-\pi}^{\pi}
{\hat{f}_m(q) \sin(n q) e^{i m \omega_0 \tau}\,  dq},
\ \ \
U_n(\tau) = \sum_{m \in \Z} \int_{-\pi}^{\pi}
{\hat{u}_m(q) \sin(n q) e^{i m \omega_0 \tau}\,  dq} .
$$
Note that the Fourier transform defines
an isometric isomorphism between the Hilbert spaces $\mathbb{X}$
and $\mathcal{X} := \ell^2 (\Z ; L^2_{\rm per}(0, 2\pi))$,
and between $\mathbb{D}$ and $\mathcal{D} := \ell^2_2 (\Z ; L^2_{\rm per}(0, 2\pi))$, where the latter denotes
the usual Hilbert space consisting of sequences $\{\hat{u}_m\}_{m \in \mathbb{Z}}$ in $L^2_{\rm per}(0, 2\pi)$
for which $\{ m^2 \hat{u}_m \} \in \ell^2 (\Z ; L^2_{\rm per}(0, 2\pi))$.

In the Fourier space, differential equations (\ref{inhomogen-problem}) reduce to the set of algebraic equations
\begin{equation}
\label{dualpb}
[\, \kappa^2 - m^2 \omega_0^2 \gamma + 2 \gamma
(1 - \cos{q})\, ]\, \hat{u}_m(q) = \hat{f}_m(q) ,
\end{equation}
where $\hat{\bf f} \in \mathcal{X}$ and we look for $\hat{\bf u} \in \mathcal{D}$.
Thanks to condition (\ref{resonances-again}),
for all $m \in \Z$, equation (\ref{dualpb})
has a unique solution $\hat{u}_m =\hat{H}_m\, \hat{f}_m \in  L^2_{\rm per}(0, 2\pi)$,
where
$$
\hat{H}_m (q)= [\, \kappa^2 - m^2 \omega_0^2 \gamma + 2 \gamma (1 - \cos{q})\, ]^{-1}
$$
is $2\pi$-periodic and analytic over $\mathbb{R}$.
To check that $\hat{\bf u} \in \mathcal{D}$, we use the estimate
\begin{equation}
\label{series}
\sum_{m \in \Z}{m^4\, \|  \hat{u}_m \|_{L^2_{\rm per}}^2 }\leq
\sum_{m \in \Z}{m^4\, \|  \hat{H}_m \|_{L^\infty}^2 \, \|  \hat{f}_m \|_{L^2_{\rm per}}^2 } \leq
\big(\sup_{m \in \Z}{(m^2\, \|  \hat{H}_m \|_{L^\infty}) \, }\big)^2\, \| \hat{\bf f} \|_{\mathcal{X}}^2  ,
\end{equation}
where
$$
\|  \hat{H}_m \|_{L^\infty} =
\frac{1}{
\inf_{q \in [-\pi,\pi]} \left| \kappa^2 - m^2
\omega_0^2 \gamma + 2 \gamma (1 - \cos(q)) \right|
}.
$$
For $|m| \geq \omega_0^{-1}(4+\gamma^{-1}\kappa^2)^{1/2}$ one has
$$
\|  \hat{H}_m \|_{L^\infty} = \frac{1}{
 m^2\omega_0^2 \gamma -\kappa^2 - 4 \gamma
}
$$
and $\sum_{m \in \Z}{m^4\, \|  \hat{u}_m \|_{L^2_{\rm per}}^2 } < \infty$
according to (\ref{series}). Consequently, there exists
a unique solution ${\bf U}\in \mathbb{D}$
of the inhomogeneous system (\ref{inhomogen-problem}) in the form
\begin{equation}
\label{solution-inhomogen-problem} U_n(\tau) = \sum_{m \in \Z} \int_{-\pi}^{\pi}
\frac{\hat{f}_m(q) \sin(n q) e^{i m \omega_0 \tau}}{\kappa^2 - m^2 \omega_0^2 \gamma + 2 \gamma
(1 - \cos{q})} dq .
\end{equation}
Moreover, we have
\begin{equation}
\label{bound-resolv}
\| {\bf U} \|_{\mathbb{X} }^2 =
\sum_{m \in \Z}{\|  \hat{u}_m \|_{L^2_{\rm per}}^2 } \leq
\sum_{m \in \Z}{\|  \hat{H}_m \|_{L^\infty}^2 \, \|  \hat{f}_m \|_{L^2_{\rm per}}^2 } \leq
\big(\sup_{m \in \Z}{\|  \hat{H}_m \|_{L^\infty} \, }\big)^2\,
\| {\bf F} \|_{\mathbb{X}}^2.
\end{equation}
This proves estimate (\ref{bound-resolvent-intermediate}). Estimate
(\ref{bound-resolvent-intermediate-2}) follows by differentiating
(\ref{inhomogen-problem}) with respect to $\tau$ and using
(\ref{bound-resolvent-intermediate}).
\end{proof}

Now we come back to the case in which $\omega_0$ depends in fact
on $\gamma$ and on the fixed energy $E$ of the excitation at site $n=0$.
More precisely, we recall that
$$
\omega_0(E,\gamma) = \frac{2\pi}{T(E,\gamma)} \gamma^{-1/2} = \sqrt{2} + {\cal O}(\gamma^{1/2}) \quad
\mbox{\rm as} \quad \gamma \to 0.
$$
Non-resonance conditions (\ref{resonances-again})
are satisfied if $\gamma$ belongs to the disjoint set
$C_{E} = \cup_{m \geq m_0} (\Gamma_m,\gamma_m)$, where $\Gamma_m$ and
$\gamma_m$ are roots of equations (\ref{Gamma-gamma}) for $m$ large enough. For each $\gamma \in C_{E}$,
a unique solution of the inhomogeneous system (\ref{inhomogen-problem}) exists in the form
(\ref{solution-inhomogen-problem}). However, the norm $\| {\bf U} \|_{\mathbb{X} }$
diverges as $\gamma$ approaches the boundary of $C_{E}$, according to estimate
(\ref{bound-resolvent-intermediate})-(\ref{constant}).
Following the approach introduced in \cite{James4}, we
quantify this divergence when $\gamma $ tends towards $0$
but remains in a well-chosen subset $\tilde{C}_{E,\nu}$ of $C_{E}$
far enough from resonances.

\begin{lemma}
\label{lemma-resolvent} Fix $E>V_L$ and $\nu \in (0,1)$. Let $\gamma \in
\tilde{C}_{E,\nu} = \cup_{m \geq m_0}
(\tilde{\Gamma}_m,\tilde{\gamma}_m) \subset C_{E}$, where
$m_0 \geq 1$ is large enough and
$\tilde{\Gamma}_m$, $\tilde{\gamma}_m$ are found from the equations
\begin{equation}
\label{Gamma-gamma-tilde-lemma}
\frac{\sqrt{\kappa^2
+ 4 \, \tilde{\Gamma}_m}}{\sqrt{(m+1)^2 - \nu (m+1)}} = \sqrt{\tilde{\Gamma}_m} \omega_0(E,\tilde{\Gamma}_m), \quad
\frac{\kappa}{\sqrt{m^2 + \nu m}} = \sqrt{\tilde{\gamma}_m} \omega_0(E,\tilde{\gamma}_m ),
\end{equation}
for $m\geq m_0$, and satisfy as $m\rightarrow +\infty$
\begin{equation}
\label{equivg}
\tilde{\Gamma}_m =
\frac{\kappa^2}{2 m^2}\Big(1 + \frac{\kappa \lambda(E)}{\pi m} - \frac{2-\nu}{m}
+ {\cal O}(m^{-2}) \Big), \quad \tilde{\gamma}_m =
\frac{\kappa^2}{2 m^2}
\Big(1 + \frac{\kappa \lambda(E)}{\pi m} - \frac{\nu}{m}  + {\cal O}(m^{-2}) \Big).
\end{equation}
There exist $\gamma_0(\nu ) > 0$ and $C_0(\nu ) > 0$ such that
for any $\gamma \in \tilde{C}_{E,\nu } \cap (0,\gamma_0(\nu))$
and any ${\bf F} \in \mathbb{X}$,
the solution ${\bf U}= \mathcal{L}^{-1}\, {\bf F}  \in \mathbb{D}$ of the inhomogeneous equation
(\ref{inhomogen-problem}) satisfies
\begin{equation}
\label{bound-resolvent-L-2}
\| {\bf U} \|_{\mathbb{X} }
\leq
C_0(\nu) \gamma^{-1/2}
\| {\bf F} \|_{\mathbb{X}}.
\end{equation}
Moreover, if in addition ${\bf F} \in H^1_{\rm per}((0,T_0);l^2(\mathbb{N}))$
\begin{equation}
\label{bound-resolvent}
\| {\bf U} \|_{ H^1_{\rm per}((0,T_0);l^2(\mathbb{N}))  } \leq
C_0(\nu) \gamma^{-1/2}
\| {\bf F} \|_{ H^1_{\rm per}((0,T_0);l^2(\mathbb{N}))  }.
\end{equation}
\end{lemma}

\begin{proof}
Equations (\ref{Gamma-gamma-tilde-lemma}) can be solved for $m$ large enough
(say $m\geq m_0 (E)$) and small $ \tilde{\Gamma}_m , \tilde{\gamma}_m $
by combining expansion (\ref{period-expansion-T0}) and the implicit function arguments.
In order to deduce estimates (\ref{bound-resolvent-L-2}) and (\ref{bound-resolvent})
from Lemma \ref{lemma-4}, we need a lower bound for
\begin{equation}
\label{infpb}
M(n,q)=\kappa^2 - n^2
\omega_0^2 \gamma + 2 \gamma (1 - \cos(q)), \quad n \in \Z, \;\; q \in [-\pi,\pi].
\end{equation}

Let us assume $\gamma \in (\tilde{\Gamma}_m,\tilde{\gamma}_m)$
with $m$ large enough. In what follows we will show that $M(m,q)>0$ and $M(m+1,q)<0$.
For $n\leq m-1$ it follows that $M(n,q)>M(m,q)>0$, and
for $n\geq m+2$ we have also $M(n,q)<M(m+1,q)<0$, hence the infimum of (\ref{infpb})
will be reached for $n=m$ or $n=m+1$.

Let us start with the case $n=m$. By combining (\ref{defw0}) and (\ref{period-expansion}) it follows that
$$
\partial_{\gamma}(\gamma \omega_0^2(E,\gamma ))=2+ {\cal O}(\sqrt{\gamma})>0
$$
for small $\gamma > 0$. This property implies that the minimum of $M(m,q)$ occurs at
$\gamma = \tilde{\gamma}_m$ and $q = 0$.
Using the definition of $\tilde{\gamma}_m$ by (\ref{Gamma-gamma-tilde-lemma}), we obtain
$$
M(m,q) \geq \kappa^2 - m^2 \omega_0^2 (E,\tilde{\gamma}_m) \tilde{\gamma}_m =  \frac{\nu \kappa^2}{m + \nu} >0.
$$

Next, for  $n=m+1$ one has for $m$ large enough
$$
\partial_{\gamma}(4\gamma - (m+1)^2 \gamma \omega_0^2(E,\gamma ))=4-2(m+1)^2+ {\cal O}(m^2 \gamma^{1/2})
= -2 m^2+ {\cal O}(m)<0 .
$$
Therefore, the maximum of $M(m+1,q)$ occurs at $\gamma = \tilde{\Gamma}_m$ and $q = \pi$, where
$$
M(m+1,q)\leq \kappa^2 - (m+1)^2 \omega_0^2 (E,\tilde{\Gamma}_m) \tilde{\Gamma}_m +4 \tilde{\Gamma}_m .
$$
Using the definition of $\tilde{\Gamma}_m$ by (\ref{Gamma-gamma-tilde-lemma}), we obtain
$$
\kappa^2 - (m+1)^2 \omega_0^2 (E,\tilde{\Gamma}_m) \tilde{\Gamma}_m +4 \tilde{\Gamma}_m
=  -\frac{\nu  (m+1) \kappa^2 \omega_0^2 (E,\tilde{\Gamma}_m)}{\omega_0^2(E,\tilde{\Gamma}_m) ((m+1)^2 - \nu (m+1))
- 4},
$$
and consequently
$$
M(m+1,q)\leq -\frac{\nu  (m+1) \kappa^2 \omega_0^2 (E,\tilde{\Gamma}_m)}{\omega_0^2(E,\tilde{\Gamma}_m) ((m+1)^2 - \nu (m+1))- 4} < -\frac{\nu   \kappa^2 }{ m+1 - \nu } <0.
$$

As a result of the above analysis, we have
$$
\inf_{n \in \Z, q \in [-\pi,\pi]} \left| \, M(n,q) \,  \right| \geq  \frac{\nu   \kappa^2 }{ m+1}.
$$
Since $\gamma \leq \tilde{\gamma}_m \leq (\kappa / m)^2$ for $m$ large enough,
we get finally for all $\gamma \in \tilde{C}_{E,\nu }$ small enough
\begin{equation}
\label{bound-denominator}
\inf_{n \in \Z, q \in [-\pi,\pi]} \left| \, M(n,q) \,  \right| \geq    \frac{\nu   \kappa }{ 2}\, \gamma^{1/2}.
\end{equation}
Estimates (\ref{bound-resolvent-L-2}) and (\ref{bound-resolvent}) follow
directly from Lemma \ref{lemma-4} and estimate (\ref{bound-denominator}).
\end{proof}

\begin{remark}
If the set $\tilde{C}_{E,\nu} = \cup_{m \geq m_0}
(\tilde{\Gamma}_m,\tilde{\gamma}_m) \subset C_{E}$ is defined by
$$
\frac{\sqrt{\kappa^2
+ 4 \, \tilde{\Gamma}_m}}{\sqrt{(m+1)^2 - \nu (m+1)^q}} = \sqrt{\tilde{\Gamma}_m}
\omega_0(E,\tilde{\Gamma}_m), \quad
\frac{\kappa}{\sqrt{m^2 + \nu m^q}} = \sqrt{\tilde{\gamma}_m} \omega_0(E,\tilde{\gamma}_m ),
$$
for some $q \in (0,2)$ and $\nu > 0$, then $\tilde{\Gamma}_m$ and $\gamma_m$ satisfy as $m\rightarrow +\infty$
$$
\tilde{\Gamma}_m =
\frac{\kappa^2}{2 m^2}\Big(1 + \frac{\kappa \lambda(E)}{\pi m} - \frac{2}{m} + \frac{\nu}{m^{2-q}}
+ {\cal O}(m^{-2}) \Big), \quad \tilde{\gamma}_m =
\frac{\kappa^2}{2 m^2}
\Big(1 + \frac{\kappa \lambda(E)}{\pi m} - \frac{\nu}{m^{2-q}}  + {\cal O}(m^{-2}) \Big).
$$
From the requirement that $\tilde{\Gamma}_m < \tilde{\gamma}_m$ for large $m \geq m_0$, we can see
that either $q \in (0,1)$ and $\nu > 0$ or $q = 1$ and $\nu \in (0,1)$.
On the other hand, we have from the proof of Lemma \ref{lemma-resolvent} that
$$
\inf_{n \in \Z, q \in [-\pi,\pi]} \left| \, M(n,q) \,  \right| \geq  C(\nu) \gamma^{(2-q)/2}.
$$
For $q \in (0,1)$, the interval $(\tilde{\Gamma}_m,\tilde{\gamma}_m)$ converges to the interval
$(\Gamma_m,\gamma_m)$ for the price of losing too much power of $\gamma$ in the bound (\ref{bound-resolvent}).
Therefore, the estimate of Lemma \ref{lemma-resolvent} is sharp in this sense.
\end{remark}

\section{Proof of Theorem \ref{theorem-main}}
\label{section-4}

Let us represent
$$
x_0(t) = \frac{1}{\gamma^{1/2}} X_0(\tau), \quad
x_{-n}(t) = x_n(t) = X_n(\tau), \;\; n \geq 1,
$$
where $\{ X_n \}_{n \in \mathbb{Z}}$ is a new set of unknowns in time $\tau = \gamma^{1/2} t$.
From the discrete Klein--Gordon equation (\ref{KGlattice}), we obtain
\begin{eqnarray}
\label{equation-X0-zero}
& \phantom{t} &
\ddot{X}_0 + 2 X_0 + \gamma^{-1/2} V'(\gamma^{-1/2} X_0) = 2 \gamma^{1/2} X_1, \\
\label{equation-X1}
& \phantom{t} & \gamma \ddot{X}_1 + \kappa^2 X_1 + N(X_1) = \gamma (X_2 - 2 X_1) + \gamma^{1/2} X_0, \\
\label{equation-Xn}
& \phantom{t} & \gamma \ddot{X}_n + \kappa^2 X_n + N(X_n) = \gamma (X_{n+1} - 2 X_n + X_{n-1}), \quad n \geq 2,
\end{eqnarray}
where $N(X) := V'(X) - \kappa^2 X$.

Let $B_{\delta} \subset H^1_e$ be a ball of small radius $\delta > 0$ centered at
$0 \in H^1_e$. By assumption (P2), $N(X) : B_{\delta} \to H^1_e$ is a $C^5$ map.
Moreover, expansion $V'(x) = \kappa^2 x + {\cal O}(x^5)$
near $x = 0$ implies the existence of $C>0$ such that for all
$\delta > 0$ small enough we have
\begin{equation}
\label{bound-on-N}
\forall X \in B_{\delta}, \quad
\| N(X) \|_{H^1_e} \leq C \| X \|^5_{H^1_e},
\end{equation}
\begin{equation}
\label{bound-on-LC}
\forall X_1, X_2 \in B_{\delta}, \quad
\| N(X_1) -N(X_2) \|_{H^1_e} \leq C \, \delta^4\, \| X_1 -X_2 \|_{H^1_e}.
\end{equation}

From system (\ref{equation-X1}) and (\ref{equation-Xn}), we can see that
oscillations near the zero solution would involve inverting the linearized
operator in the inhomogeneous system (\ref{inhomogen-problem}). By estimate (\ref{bound-resolvent}),
we are going to loose $\gamma^{1/2}$, which is the size of the inhomogeneous term $\gamma^{1/2} X_0$.
This would prevent us from using the contraction mapping theorem in the neighborhood of the zero solution.
To overcome this obstacle, we introduce the near-identity transformation
$$
X_1 = Y_1 + \gamma^{1/2} \kappa^{-2} X_0, \quad
X_n = Y_n, \;\; n \geq 2
$$
and rewrite system (\ref{equation-X0-zero})--(\ref{equation-Xn}) in the equivalent form
\begin{eqnarray}
\label{equation-X0}
& \phantom{t} &
\ddot{X}_0 + 2 X_0 + \gamma^{-1/2} V'(\gamma^{-1/2} X_0) = 2 \gamma \kappa^{-2} X_0 + 2 \gamma^{1/2} Y_1, \\
\label{equation-X1-new}
& \phantom{t} &
\gamma \ddot{Y}_1 + \kappa^2 Y_1 -\gamma (Y_2 - 2 Y_1)
+ N(Y_1 + \gamma^{1/2} \kappa^{-2} X_0) =
 - \gamma^{3/2} \kappa^{-2} \left( \ddot{X}_0 + 2 X_0 \right),\\
\label{equation-X2-new}
& \phantom{t} &
\gamma \ddot{Y}_2 + \kappa^2 Y_2 - \gamma (Y_3 - 2 Y_2 + Y_1)
+ N(Y_2) =   \gamma^{3/2} \kappa^{-2} X_0, \\
\label{equation-Xn-new}
& \phantom{t} &
\gamma \ddot{Y}_n + \kappa^2 Y_n
-\gamma (Y_{n+1} - 2 Y_n + Y_{n-1})
+ N(Y_n)  = 0, \quad n \geq 3.
\end{eqnarray}
Extracting $\ddot{X}_0 + 2 X_0$ from (\ref{equation-X0}), we can rewrite (\ref{equation-X1-new})
in the equivalent form
\begin{eqnarray}
\nonumber
& \gamma \ddot{Y}_1 + \kappa^2 Y_1 -\gamma (Y_2 - 2 Y_1)
+ N(Y_1 + \gamma^{1/2} \kappa^{-2} X_0) \\
\label{equation-X1-new-new} & \phantom{texttexttexttexttexttext} =
- 2 \gamma^2 \kappa^{-2} Y_1 - 2 \gamma^{5/2} \kappa^{-4} X_0
+ \gamma \kappa^{-2} V'(\gamma^{-1/2} X_0).
\end{eqnarray}
We shall solve the above system in two steps, using the contraction mapping theorem
to solve (\ref{equation-X2-new})-(\ref{equation-X1-new-new}) at fixed $X_0$, and
then Schauder's fixed point theorem to solve (\ref{equation-X0}). In the latter case,
we shall consider equation (\ref{equation-X0}) similar to equation (\ref{source-perturbation})
with $X_0 \in H^1_e$ being close to the solution $X \in H^1_e$
of Lemma \ref{lemma-oscillations} rescaled by
(\ref{scaling-transformation}). The source term depends on $Y_1 \in B_{\delta} \subset H^1_e$ and $X_0$,
and it will be proved that $\delta = {\cal O}(\gamma^{\varepsilon})$ is small as $\gamma \to 0$.

Let us now describe our functional setting in more detail.
Given $\mu \in (0,\frac{1}{2})$ and
$\gamma >0$ small enough, we define
\begin{equation}
\label{neighborhood}
D_{\mu ,\gamma} = \left\{ X_0 \in H^1_e \cap C^1_e : \quad \| X_0  \|_{C^1}
\leq 3 \sqrt{E} ,
\quad
X_0(\tau ) \geq a_0 \sqrt{\gamma} \mbox{ for } 0\leq \tau \leq \frac{T_0}{4}- \gamma^{\frac{1}{2} -\mu}  \right\} .
\end{equation}
When $\gamma$ is small enough, Corollary \ref{lemma-H2-solution} and
Lemma \ref{estimdeltat0} imply that $X \in D_{\mu ,\gamma}$.
Moreover, $D_{\mu ,\gamma}$ defines a closed, bounded and convex subset of $C^1_{\rm per}(0,T_0)$.
Repeating the same arguments as in the proof of
Corollary \ref{lemma-H4-solution}, we obtain
\begin{equation}
\label{bound-on-X-0}
\exists C > 0 : \quad \forall X_0 \in D_{\mu ,\gamma} : \quad
\| V'(\gamma^{-1/2} X_0) \|_{H^1_e} \leq C \gamma^{-\frac{1}{4}-\frac{\mu}{2}}.
\end{equation}

For sufficiently small $\gamma$ in the set
$\tilde{C}_{\omega_0,\nu}$ for fixed $\nu \in (0,1)$, we can rewrite system (\ref{equation-X2-new}),
(\ref{equation-Xn-new}), and (\ref{equation-X1-new-new}) in the form
$$
{\bf Y} + {\cal L}^{-1} {\bf N}({\bf Y},X_0) = {\cal L}^{-1} {\bf F}(Y_1,X_0),
$$
where ${\cal L}^{-1}$ is the Green operator of Lemma \ref{lemma-resolvent} solving the linear inhomogeneous problem
(\ref{inhomogen-problem}),
$$
{\bf N}({\bf Y},X_0) : H^1_e((0,T_0);l^2(\mathbb{N})) \times D_{\mu ,\gamma} \rightarrow H^1_e((0,T_0);l^2(\mathbb{N}))
$$
is the nonlinear operator at the left side of (\ref{equation-X2-new})-(\ref{equation-X1-new-new}), and
$$
{\bf F}(Y_1,X_0) : B_{\delta} \times D_{\mu ,\gamma} \rightarrow H^1_e((0,T_0);l^2(\mathbb{N}))
$$
is the right side of (\ref{equation-X2-new})-(\ref{equation-X1-new-new}).
In order to use the estimates of Lemma
\ref{lemma-resolvent}, we assume $\gamma$ small enough in $\tilde{C}_{\omega_0,\nu}$.
Using (\ref{bound-resolvent}) and (\ref{bound-on-X-0}), we obtain that
\begin{equation}
\label{inhomogeneous-term}
\exists M > 0 : \quad \forall Y_1 \in B_{1}, \quad
\forall X_0 \in D_{\mu ,\gamma} : \quad
\| {\cal L}^{-1} {\bf F}(Y_1,X_0) \|_{H^1_e((0,T_0);l^2(\mathbb{N}))} \leq \frac{M}{2} \gamma^{\epsilon},
\end{equation}
where $\epsilon = \frac{1}{4}-\frac{\mu}{2}$. Now let us
denote by $\mathbb{B}_\delta$
the ball of radius $\delta = M\, \gamma^{\epsilon}$
centered at $0$ in $H^1_e((0,T_0);l^2(\mathbb{N}))$.
Using (\ref{bound-resolvent}) and (\ref{bound-on-N}), we obtain
\begin{equation}
\label{nl-term}
\exists C > 0 : \quad \forall {\bf Y} \in \mathbb{B}_{\delta}, \quad
\forall X_0 \in D_{\mu ,\gamma} : \quad
\| {\cal L}^{-1} {\bf N}({\bf Y},X_0) \|_{H^1_e((0,T_0);l^2(\mathbb{N}))} \leq C \gamma^{5 \epsilon-\frac{1}{2}}.
\end{equation}

Let us further assume $\mu \in \left(0, \frac{1}{4}\right)$, which implies $\epsilon \in \left(\frac{1}{8},\frac{1}{4}\right)$.
From (\ref{inhomogeneous-term})-(\ref{nl-term}) and the triangle inequality, it follows that the map
${\cal L}^{-1} ({\bf N}-{\bf F})(\cdot,X_0)$ maps $\mathbb{B}_\delta$
into itself for $\gamma$ small enough in $\tilde{C}_{\omega_0,\nu}$
and for all $X_0 \in D_{\mu ,\gamma}$.
Moreover, thanks to bound (\ref{bound-on-LC}) and Lemma
\ref{lemma-resolvent}, for all sufficiently small $\gamma$ in
$\tilde{C}_{\omega_0,\nu}$ and for all $X_0 \in D_{\mu ,\gamma}$, the map
${\cal L}^{-1} ({\bf N}-{\bf F})(\cdot,X_0)$ is a contraction in $\mathbb{B}_\delta$, with Lipschitz constant
$ {\cal O}(\gamma^{4\epsilon - \frac{1}{2}})$.
By the contraction mapping theorem (and using the fact that ${\bf N}$ and ${\bf F}$ are in addition
locally Lipschitzian with respect to $X_0 \in H^1_e$), there exists a unique continuous map
\begin{equation}
\label{continuous-map}
D_{\mu ,\gamma} \ni X_0 \mapsto {\bf Y} \in H^1_e((0,T_0);l^2(\mathbb{N}))
\end{equation}
such that $\{ Y_n \}_{n \geq 1}$ solves
(\ref{equation-X2-new})-(\ref{equation-X1-new-new}) and satisfies the bound
\begin{equation}
\label{bound-on-Y}
\exists M > 0 : \quad \forall X_0 \in D_{\mu ,\gamma} : \quad
\| {\bf Y} \|_{H^1_e((0,T_0);l^2(\mathbb{N}))} \leq M  \gamma^{\epsilon}.
\end{equation}

We can now substitute $Y_1$ from solutions of
system (\ref{equation-X2-new})-(\ref{equation-X1-new-new}) to equation (\ref{equation-X0}).
Applying Lemma \ref{lemma-spectrum}, we rewrite
equation (\ref{equation-X0}) in the form
\begin{equation}
\label{fixed-point-equation}
X_0  =  {\cal F}_{\gamma , \mu}(X_0),
\end{equation}
where ${\cal F}_{\gamma , \mu}\, : D_{\mu ,\gamma}  \rightarrow C^1_e$ is defined by
$$
{\cal F}_{\gamma , \mu}(X_0) = \mathcal{G}_{\gamma , \epsilon}(\gamma^{\frac{1}{2}-\epsilon} 2 \kappa^{-2} X_0 +
2 Y_1(X_0) \gamma^{-\epsilon}),
$$
$\mathcal{G}_{\gamma , \epsilon}$ is the nonlinear Green operator of Lemma \ref{lemma-spectrum}
solving equation (\ref{source-perturbation}), and
$Y_1(X_0)$ is defined from the map (\ref{continuous-map}).
By Lemma \ref{lemma-continuity} and the continuity of the map (\ref{continuous-map}),
the map ${\cal F}_{\gamma , \mu}$ is continuous.
Moreover, thanks to the estimates of Lemma \ref{lemma-spectrum} and the fact that $\mu >0$,
${\cal F}_{\gamma , \mu}$ maps $D_{\mu ,\gamma}$ into itself when $\gamma$ is small enough.
Observing that $\Delta = \frac{d^2}{d\tau^2}\, : H^3_e \rightarrow H^1_e$ is invertible
and considering equation (\ref{source-perturbation}) that defines $\mathcal{G}_{\gamma , \epsilon}(F)$
implicitly, we have the equality
\begin{equation}
\label{regul}
G_{\gamma , \epsilon}(F)=\Delta^{-1}\, [\,  \gamma^{\varepsilon+1/2} F -2 G_{\gamma , \epsilon}(F)
-  \gamma^{-1/2} V'(\gamma^{-1/2}  G_{\gamma , \epsilon}(F)    )  \, ] .
\end{equation}
Since the embedding of $H^3_e$ into $C^1_e$ is compact,
it follows that $G_{\gamma , \epsilon}\, : \, B_\delta \subset H^1_e \rightarrow C^1_e$
is compact, hence ${\cal F}_{\gamma , \mu}\, : D_{\mu ,\gamma}  \rightarrow C^1_e$ is compact.
Consequently, by the Schauder fixed-point theorem,
there exists a solution $X_0 \in D_{\mu ,\gamma}$ of equation (\ref{fixed-point-equation})
for sufficiently small $\gamma > 0$.
Moreover, Lemma \ref{lemma-spectrum} ensures the existence of $\theta >0$ such that
$$
X_0(\tau ) \geq a_0 \sqrt{\gamma} \quad \mbox{ for } \quad
0\leq \tau \leq \frac{T_0}{4}- \theta \gamma^{\frac{1}{2}}.
$$
Repeating the same estimates as above with $\mu=0$, we obtain
\begin{equation}
\label{bound-on-Yimproved}
\| {\bf Y} \|_{H^1_e((0,T_0);l^2(\mathbb{N}))} = {\cal O}(\gamma^{1/4}).
\end{equation}
Fixing now $\epsilon = \frac{1}{4}$ in Lemma \ref{lemma-spectrum},
estimate (\ref{bound-L2-F}) yields finally
\begin{equation}
\label{bound-on-X0}
\exists C > 0 : \quad \| X_0 - X \|_{H^1_e} \leq C \gamma^{\frac{1}{4}}.
\end{equation}
Combining all transformations above with
bounds (\ref{bound-on-Yimproved}) and (\ref{bound-on-X0}) as well as using
embedding of $H^1_e$ into $L^{\infty}_e$ and of $l^2(\mathbb{N})$
into $l^{\infty}(\mathbb{N})$, we obtain bound (\ref{bound-final}).

\begin{remark}
If assumption (P2) is relaxed with the expansion $V'(x) = \kappa^2 x + {\cal O}(x^3)$ near $x = 0$,
then the map ${\cal L}^{-1} {\bf N}({\bf Y},X_0)$ is a contraction operator
with respect to ${\bf Y}$ in a ball of radius $\delta = \gamma^{\varepsilon}$
for any $\varepsilon > \frac{1}{4}$. In this case, the inhomogeneous term ${\cal L}^{-1} {\bf F}(Y_1,X_0)$
with the bound (\ref{inhomogeneous-term}) is critical with $\varepsilon = \frac{1}{4}$
and prevent us to close the arguments of the contraction mapping theorem.
\end{remark}

\begin{remark}
The Lipschitz continuity of the map (\ref{continuous-map})
can also be established but the Lipschitz constant may have a bad behavior
as $\gamma \to 0$ because of the factor $\gamma^{-1/2}$ in the last term
$V(\gamma^{-1/2} X_0)$ of equation (\ref{equation-X1-new-new}). This
obstacle prevents us from the use of the contraction mapping theorem
for equation (\ref{fixed-point-equation}).
\end{remark}

\section{Exponential decay on $\mathbb{Z}$}
\label{section-5}

Our last result is to show that the large-amplitude discrete breather
constructed in Theorem \ref{theorem-main} decays exponentially in $n$ on $\mathbb{Z}$.
The arguments repeat those of reference \cite{James4}, to which we shall refer for
some standard steps of the proof.

\begin{lemma}
Let ${\bf x}(t) \in l^2(\mathbb{Z},L^{\infty}_{\rm per}(0,T))$ be the solution in Theorem \ref{theorem-main}.
There exists a constant $D_0 > 0$ such that
\begin{equation}
\label{exponential-decay}
\sup_{t \in [0,T]} | x_n(t) | \leq (D_0 \gamma)^{(2n-1)/4}, \quad n \geq 2,
\end{equation}
for all sufficiently small $\gamma  \in \tilde{C}_{E,\nu}$.
\end{lemma}

\begin{proof}
The operator $ \mathcal{A}_\gamma = \gamma \frac{d^2}{d\tau^2}+\kappa^2
\, : H^3_e \subset H^1_e \rightarrow H^1_e$ is unbounded, closed
and self-adjoint. Its spectrum consists of simple eigenvalues
$\kappa^2-\gamma k^2 {\omega}_0^2$ ($k \in \mathbb{Z}$), hence
$$
\| \mathcal{A}_\gamma^{-1} \|_{\mathcal{L}(H^1_e) }=
\frac{1}{ \inf_{k \in \Z}{| \kappa^2-\gamma k^2 {\omega}_0^2 |}}
=
\frac{1}{ \inf_{k \in \Z}{| M(k,0) |}},
$$
where $M(k,q)$ is defined in (\ref{infpb}). Using estimate
(\ref{bound-denominator}) we get consequently
\begin{equation}
\label{invop}
\| \mathcal{A}_\gamma^{-1} \|_{\mathcal{L}(H^1_e) }= {\cal O}(\gamma^{-1/2})
\end{equation}
when $\gamma \rightarrow 0$ in $\tilde{C}_{E,\nu}$.

Now we rewrite system (\ref{equation-Xn-new}) as
\begin{equation}
\label{equation-Xn-newbis}
Y_n = \mathcal{A}_\gamma^{-1}\, [\,   \gamma (Y_{n+1} - 2 Y_n + Y_{n-1})
- N(Y_n)  \, ],
 \quad n \geq 3.
\end{equation}
By estimate (\ref{bound-on-Yimproved}) we have $\| Y_n \|_{H^1_e} = {\cal O}(\gamma^{1/4})$
uniformly in $n \in \N$, which in conjunction with (\ref{bound-on-N}) yields
$\| N(Y_n) \|_{H^1_e} \leq C \, \gamma\, \| Y_n \|_{H^1_e}$.
Using this estimate and the bound (\ref{invop}) in (\ref{equation-Xn-newbis}), one finds
$M>0$ such that for all $\gamma \in \tilde{C}_{E,\nu}$ small enough and for all $n\geq 3$
\begin{equation}
\label{est2}
 \| \, Y_n  \, \|_{H^1_e} \leq
M\, \gamma^{1/2}\,
(\,
 \| \, Y_{n+1}  \, \|_{H^1_e}
+ \, \| \, Y_{n-1}  \, \|_{H^1_e}
\, ) .
\end{equation}
A simple application of the discrete maximum principle yields then
(see \cite{James4}, Lemma 3.3)
\begin{equation}
\label{estiminter}
\| Y_n \|_{H^1_{e}} \leq (2 M \gamma^{1/2})^{n-2}\,  \| \, Y_2  \, \|_{H^1_e}, \quad n \geq 2.
\end{equation}
Using equation (\ref{equation-X2-new}), estimates (\ref{invop}) and (\ref{estiminter}) with $n=3$,
the fact that $\| Y_1 \|_{H^1_e} = {\cal O}(\gamma^{1/4})$ and
$\| X_0 \|_{H^1_e} = {\cal O}(1)$
(direct consequence of (\ref{bound-on-X0}) and Corollary \ref{lemma-H2-solution}), we get
\begin{equation}
\label{estimy2}
\| Y_2 \|_{H^1_e} = {\cal O}(\gamma^{3/4}).
\end{equation}
Then one completes the proof by
putting estimates (\ref{estiminter}) and (\ref{estimy2}) together and using
the continuous embedding of $H^1_e$ in $L^{\infty}_e$.
\end{proof}


\begin{thebibliography}{99}

\bibitem{archistab}
J.F.R. Archilla, J. Cuevas, B. S\'anchez-Rey and A. Alvarez, {``Demonstration of the stability or instability of multibreathers at low coupling"},
{Physica D} {\bf 180} (2003), 235-255.

\bibitem{Aubry2} S. Aubry, ``Breathers in nonlinear lattices: Existence, linear stability and
quantization", Physica D {\bf 103} (1997), 201--250.

\bibitem{aubkadel}
S. Aubry, G. Kopidakis and V. Kadelburg, {``Variational proof for hard discrete breathers in some
classes of Hamiltonian dynamical systems"}, {Discrete and Continuous Dynamical Systems B {\bf 1}} (2001), 271-298.

\bibitem{Bambusi} D. Bambusi, ``Exponential stability of breathers in Hamiltonian networks of weakly
coupled oscillators", Nonlinearity {\bf 9} (1996), 433--457.

\bibitem{Fura} J. Fura and S. Rybicki, ``Periodic solutions of second order Hamiltonian systems bifurcating from
infinity", Annales de l'Institut Henri Poincar\'{e} (C) Analyse
Non Lin\'{e}aire {\bf 24} (2007), 471--490.

\bibitem{james}
G. James, {``Centre manifold reduction for quasilinear discrete systems"}, J.\ Nonlinear Sci. {\bf 13} (2003), 27-63.

\bibitem{James4} G. James, A. Levitt, and C. Ferreira, ``Continuation of discrete breathers
from infinity in a nonlinear model for DNA breathing'', Applicable
Analysis {\bf 89} (2010), 1447--1465.

\bibitem{James2} G. James, B. S\'anchez--Rey, and J. Cuevas, ``Breathers in inhomogeneous nonlinear
lattices: an analysis via center manifold reduction'', Rev. Math. Phys. {\bf 21} (2009), 1--59.

\bibitem{koukou}
V. Koukouloyannis and P. Kevrekidis,
``On the stability of multibreathers in Klein-Gordon chains",
Nonlinearity {\bf 22} (2009), 2269--2285.

\bibitem{MA94} R.S. MacKay and S. Aubry, ``Proof of existence of breathers for time-reversible
or Hamiltonian networks of weakly coupled oscillators",
Nonlinearity {\bf 7} (1994) 1623--1643.

\bibitem{macsep}
R.S. MacKay and J-A. Sepulchre, ``Stability of discrete breathers", {Physica D} {\bf 119} (1998), 148-162.

\bibitem{marinstab}
{J.L. Marin and S. Aubry}, {``Finite size effects on instabilities of discrete breathers"}, {Physica D} {\bf 119} (1998), 163-174.

\bibitem{Mielke} A. Mielke and C. Patz, ``Dispersive stability of infinite-dimensional Hamiltonian
systems on lattices'', Applicable Analysis {\bf 89} (2010), 1493--1512.

\bibitem{pankovb}
A. Pankov, {``Travelling waves and periodic oscillations in Fermi-Pasta-Ulam lattices"},
Imperial College Press, London (2005).

\bibitem{James1} M. Peyrard, S. Cuesta--L\'opez, and G. James, ``Modelling DNA at the mesoscale: a challenge
for nonlinear science?'', Nonlinearity {\bf 21}, T91--T100 (2008)

\bibitem{James3} M. Peyrard, S. Cuesta--L\'opez, and G. James, ``Nonlinear analysis of the dynamics
of DNA breathing'', J. Biol. Phys. {\bf 35}, 73--89 (2009)

\bibitem{pey04}
M. Peyrard, {``Nonlinear dynamics and statistical physics of DNA"},
Nonlinearity {\bf 17} (2004), R1-R40.

\bibitem{sepmac}
J.-A. Sepulchre and R.S. MacKay, {``Localized oscillations in conservative or
dissipative networks of weakly coupled autonomous oscillators"},
Nonlinearity {\bf 10} (1997), 679-713.

\bibitem{Treschev} D. Treschev, ``Travelling waves in FPU
lattices", Discrete and Continuous Dynamical Systems A {\bf 11}
(2004), 867--880.

\bibitem{weber} G. Weber, ``Sharp DNA denaturation due to solvent interaction",
Europhys. Lett.  {\bf 73} (2006), 806.


\end{thebibliography}
\end{document}